%% file: main.tex
\documentclass[conference]{IEEEtran}
\IEEEoverridecommandlockouts
\usepackage{cite}
\usepackage{textcomp}
\usepackage{xcolor}
\input{extra-packages-commands.tex}

\usepackage{tabularx}

\newcommand\blfootnote[1]{
  \begingroup
  \renewcommand\thefootnote{}\footnote{#1}
  \addtocounter{footnote}{-1}
  \endgroup
}

\begin{document}

% \title{Asynchronous Weight Reassignment in Weighted Majority Quorums}
\title{How Hard is Asynchronous Weight Reassignment?
{\Large (Extended Version)}
}

\author{}
\author{
\IEEEauthorblockN{
    Hasan Heydari,\IEEEauthorrefmark{1}
    Guthemberg Silvestre,\IEEEauthorrefmark{2} and 
    Alysson Bessani\IEEEauthorrefmark{1}} 

    \IEEEauthorblockA{\IEEEauthorrefmark{1}LASIGE, Faculdade de Ciências, Universidade de Lisboa, Portugal}
    \IEEEauthorblockA{\IEEEauthorrefmark{2}ENAC, University of Toulouse, France}

    hheydari@ciencias.ulisboa.pt, silvestre@enac.fr, anbessani@ciencias.ulisboa.pt
}

\maketitle

\begin{abstract}
The performance of distributed storage systems deployed on wide-area networks can be improved using weighted (majority) quorum systems instead of their regular variants due to the heterogeneous performance of the nodes. 
A significant limitation of weighted majority quorum systems lies in their dependence on static weights, which are inappropriate for systems subject to the dynamic nature of networked environments. 
To overcome this limitation, such quorum systems require mechanisms for reassigning weights over time according to the performance variations. 
We study the problem of node weight reassignment in asynchronous systems with a static set of servers and static fault threshold. We prove that solving such a problem is as hard as solving consensus, i.e., it cannot be implemented in asynchronous failure-prone distributed systems.    
This result is somewhat counter-intuitive, given the recent results showing that two related problems -- replica set reconfiguration and asset transfer -- can be solved in asynchronous systems.
Inspired by these problems, we present two versions of the problem that contain restrictions on the weights of servers and the way they are reassigned.
We propose a protocol to implement one of the restricted problems in asynchronous systems. 
As a case study, we construct a dynamic-weighted atomic storage based on such a protocol.
We also discuss the relationship between weight reassignment and asset transfer problems and compare our dynamic-weighted atomic storage with reconfigurable atomic storage.
\end{abstract}

\blfootnote{This is the extended version of a paper to appear at the 43rd IEEE International Conference on Distributed Computing Systems (ICDCS 2023).}


\section{Introduction}
In the era of cloud computing, cryptocurrencies, and the internet of everything, distributed storage systems are required more than ever due to their fault tolerance and high availability.
Ensuring consistency of distributed storage systems is a fundamental challenging problem in distributed computing.
One well-known solution for such a problem is utilizing quorum systems~\cite{quorumSystems}.
A quorum system is a collection of sets called \textit{quorums} such that each quorum is a subset of servers, and every two quorums \textit{intersect}.
Although many types of quorum systems exist, such as grids \cite{loadCapacityAvaiQS} and trees~\cite{treeQS}, most practical distributed storage systems (e.g.,~\cite{etcd, zookeeper, fab, gpfs}) utilize the regular majority quorum system ({\small MQS}) due to its simplicity and optimal fault tolerance.

In {\small MQS}, every quorum consists of a strict majority of servers.
Although {\small MQS} is simple and optimally fault-tolerant, it might be subject to poor latency and low throughput due to practical considerations such as replica heterogeneity~\cite{readWriteQS}.
To take such considerations into account, one can use the weighted majority quorum system ({\small WMQS}), in which each server is assigned a weight (a.k.a. vote or voting power) in accordance with its access latency or request processing capacity (throughput), as determined by a monitoring system~\cite{reliabilityVoting, aware}, and the assigned weights are used to determine whether a subset of servers constitutes a (weighted) quorum.

A significant limitation of the {\small WMQS} is its reliance on static weights, which are inappropriate for dynamic systems, where servers' performance might change over time~\cite{aware, hheydari}. \gs{it is not sure if it is a limitation of the previous work or if it has already been considered previously (eventually by unfruitful attempts), in this case, not completely new. I mean, this sentence could be rephrased to highlight the main novelty of this work.}
To overcome such a limitation, {\small WMQS} can be integrated with weight reassignment protocols for changing server weights over time according to performance variations.

The main goal of this paper is to study weight reassignment in an asynchronous system with a static set of servers and static fault threshold, where an available weighted quorum is guaranteed to exist.
To this end, as a first step, we formally define the \textit{weight reassignment} problem by which weight reassignment requests can be issued and processed.
We then prove that consensus can be reduced to the weight reassignment problem, i.e., a solution to the weight reassignment problem can be used to solve consensus.
Consequently, \emph{the weight reassignment problem cannot be implemented in asynchronous failure-prone systems}.

To cope with such an impossibility, we introduce a restricted version of weight reassignment called \textit{pairwise weight reassignment}, in which the reassignments can only be done in a pairwise way.
More precisely, in the pairwise weight reassignment, the total weight of servers remains constant, and a server gains a weight $\Delta$ if and only if another server loses $\Delta$.
Reassigning weights in such a way is similar to transferring assets in $1$-asset transfer.\footnote{In the $1$-asset transfer problem, there are some accounts, each of which is owned by a server; each server can transfer some of its assets to another server if its balance does not become negative.}
Somewhat surprisingly, although $1$-asset transfer can be implemented in asynchronous failure-prone systems~\cite{cnCrypto}, we show that this is not the case for pairwise weight reassignment.

We further restrict the pairwise weight reassignment problem by mainly considering a restriction on the possible range of weights, naming it \textit{restricted pairwise weight reassignment}.
We show that such a restricted variant of the problem can be implemented in asynchronous failure-prone systems.
As a case study, we construct a dynamic-weighted atomic storage incorporating a protocol solving this variant.
    
Our dynamic-weighted atomic storage is somewhat similar to reconfigurable atomic storage,\footnote{Reconfigurable atomic storage implements atomic storage in systems with the possibility of changing the set of servers over time.} which can be implemented in asynchronous systems~\cite{dynAtomicStorageWithoutCons, effModConsensus-freeFST, jehl2017case, smartMerge, spiegelman2017dynamic}.
We further elaborate on similarities between these storage systems, discussing why the techniques used to implement reconfigurable atomic storage, e.g., generalized lattice agreement~\cite{faleiro2012generalized}, cannot be used to implement dynamic-weighted atomic storage.
    
\vspace{0.5em}
\noindent\textbf{Contributions}. The contributions of this paper are:

\begin{itemize}
    \item We formalize the \emph{weight reassignment problem} for systems with a static set of servers and fault threshold and prove it cannot be solved in asynchronous failure-prone systems.
    \item We introduce a restricted version of the problem called the \emph{pairwise weight reassignment}, in which voting power is transferred between pairs of servers.
    Although similar to asset transfer, we show that this variant cannot be implemented in asynchronous failure-prone systems.
    \item We further restrict the problem to a \emph{restricted pairwise weight reassignment} variant, which can be implemented in asynchronous failure-prone systems, and use it to build a dynamic-weighed atomic storage.
    \item We discuss the relationship between the \emph{pairwise weight reassignment} with the asset transfer problem and compare our dynamic-weighted atomic storage with reconfigurable atomic storage.
\end{itemize}

\vspace{0.5em}    
\noindent\textbf{Organization of the paper}.
Section~\ref{sec:prelim} presents our system model and preliminary definitions.
In Section~\ref{sec:abstraction}, we introduce the weight reassignment problem.
The impossibility of implementing the weight reassignment problem in asynchronous failure-prone systems is presented in Section~\ref{sec:impossibility}.
Section~\ref{sec:weak} presents the restricted versions of the weight reassignment problem: the pairwise weight reassignment and the restricted pairwise weight reassignment.
We also show that pairwise weight reassignment cannot be implemented in asynchronous failure-prone systems.
Section~\ref{sec:implem} implements the restricted pairwise weight reassignment, while in Section~\ref{ch:epoch-less-wr}, we outline a dynamic-weighted atomic storage using this implementation.
Sections~\ref{sec:discussion} and~\ref{sec:conclusion} present related work and concludes the paper, respectively.

\section{Preliminaries}\label{sec:prelim}
\vspace{0.5em}
\noindent\textbf{System model}. 
We consider an asynchronous message-passing system composed of two non-overlapping sets of processes -- a finite set of $n$ servers $\mathcal{S}$ and an infinite set of clients $\Pi$.
Every client or server knows the set of servers.
At most $f$ servers can crash, while any number of clients may crash.
A process is called \textit{correct} if it is not crashed.
Each pair of processes is connected by a reliable communication link.
Processes are sequential, i.e., a process never invokes a new operation before obtaining a response from a previous one.
In the definitions and proofs\gs{would it be necessary to refer to the appendix here?}, we make the standard assumption of the existence of a global clock not accessible to the processes.

\vspace{0.5em}
\noindent\textbf{Weighted majority quorums}. 
Our work relies on weighted majority quorums~\cite{weightedVoting,geoReplicatedSMR,howToAssignWeights,dynVotingAlgMaintainingConsistency}, so we present their definition and one of their properties that plays an essential role throughout this paper.

\begin{definition}[Weighted Majority Quorum System]
    % The weighted majority quorum system ({\small WMQS}) is the one in which every quorum consists of a set of servers whose total weight is greater than half of the total weight of all servers.
    The weighted majority quorum system ({\small WMQS}) refers to a set of quorums where each quorum is composed of a set of servers whose total weight is greater than half of the total weight of all servers.
\end{definition}

Since some minority of servers might have the majority of weights, proportionally smaller quorums can be constituted in {\small WMQS} in contrast to {\small MQS}, yielding\gs{I'd replace ``leading to'' (commonly related to a bad result) by ``yielding''.} to performance improvements.
To guarantee the availability of a distributed system based on {\small WMQS}, a relationship between the servers' weights and $f$ must be satisfied; otherwise, more than half of the voting power can be assigned to $f$ or fewer servers~\cite{geoReplicatedSMR}. 
The following property determines such a relationship.

\begin{property}[Availability of {\small WMQS}]\label{pro:availability-wmqs}
    A {\small WMQS} is available if the sum of the $f$ greatest weights is less than half of the total weight of all servers.
\end{property}

\vspace{0.5em}
\noindent\textbf{Consensus}. 
This paper shows the impossibility of solving two problems related to weight reassignment in asynchronous failure-prone systems.
Our method of showing these results involves reducing consensus to each of these problems (comprehensive explanation of these reductions detailed in Sections~\ref{sec:impossibility} and \ref{sec:weak}). 
% We use two methods to show that a problem $\mathcal{P}$ cannot be implemented in asynchronous failure-prone systems.
% The first method is based on a reduction from consensus to $\mathcal{P}$\gs{please consider rephrasing this sentence (which problem?).}.
% We say that a problem $\mathcal{P}_1$ can be \emph{reduced} to another problem $\mathcal{P}_2$ if there is a distributed algorithm $\mathcal{A}_{\mathcal{P}_2 \rightarrow \mathcal{P}_1}$ that can take the output of $\mathcal{P}_2$ and compute $\mathcal{P}_1$’s output.
In consensus, each correct process proposes a value $v$ through the invocation of $\mathtt{propose}(v)$, and the ultimate goal of processes is to decide upon a single value from the proposed values. 
The operation must return the decided value, and any algorithm that solves consensus must satisfy the following properties (e.g.,~\cite{fischer1983consensus}):
\begin{itemize}    
    \item Agreement. All correct processes must decide the same value.
    \item Validity. If all correct processes propose the same value $v$, they must decide $v$.
    \item Termination. All correct processes must eventually decide.
\end{itemize}

% The second method is based on a notion called \textit{consensus number}.
% Each concurrent type $T$ has an associated consensus number~\cite{wait-free}, which is the maximal number of processes that can solve consensus using objects of type $T$ and other objects implemented in asynchronous failure-prone systems, like atomic register~\cite{cnCrypto, wait-free}.
% For instance, the consensus number of queue is two~\cite{wait-free}, which means that two processes can solve consensus using queue objects and atomic registers while three cannot.
% Since consensus cannot be implemented in asynchronous failure-prone systems~\cite{fLP}, by showing that at least two processes using objects of type $T$ and other objects implemented in asynchronous failure-prone systems can solve consensus, one can conclude that an object of type $T$ cannot be implemented in asynchronous failure-prone systems.

\section{Weight Reassignment Problem}\label{sec:abstraction}
The \emph{weight reassignment problem} aims to capture the safety and liveness properties that must be satisfied as servers' weights change over time.
To formalize this problem, we first need to define the $\mathit{change}$ data structure, which contains the essential information related to the outcome of a weight reassignment operation.

\vspace{0.5em}
\noindent\textbf{Change.} 
Each process $p_i$ (server or client) has a local counter denoted by $\mathit{lc}_{i}$.
We define $\mathit{change}$ as $\{ \mathcal{S} \cup \Pi \} \times \mathbb{N} \times \mathcal{S} \times \mathbb{R}$, where the quadruple $\langle p_i, \mathit{lc}_{i}, s, \Delta \rangle$ indicates that the weight of server $s$ is changed by $\Delta$ as an outcome of a reassignment request made by process $p_i$ with local counter $\mathit{lc}_{i}$.
For any change $\langle *, *, s, \Delta \rangle$, by convention, ``the weight of the change'' refers to $\Delta$, and we say that ``the change is created for $s$''.

\vspace{0.5em}
\noindent\textbf{The problem.}
We introduce a weight reassignment problem with the following operations: 

\begin{itemize}
    \item $\mathtt{reassign}(s, \Delta)$, where $s$ is a server, and $\Delta$ is a real number different from zero, and
    \item $\mathtt{read\_changes}(s)$, where $s$ is a server.
\end{itemize}

Each process can invoke $\mathtt{reassign}(s, \Delta)$ to request changing the weight of server $s$ by $\Delta$.
Without loss of generality, we assume that \emph{only servers} invoke $\mathtt{reassign}$.
When an invocation of $\mathtt{reassign}$ is \textit{completed} (see definition below), a change $c$ corresponding to the invocation's outcome is created, and a message $\langle \mathtt{Complete}, c \rangle$ is returned to the server that invoked the operation.
Any process can invoke $\mathtt{read\_changes}$ to learn about the set of changes created for a server, by which the weight of the server can be calculated. 
Each process must increment its local counter after each invocation of the $\mathtt{reassign}$ operation.

\begin{definition}[\textit{Completed} $\mathtt{reassign}$]\label{def:best-value}
    Assume that a server $s_i$ invokes $\mathtt{reassign}(s_j, \Delta)$ when its local counter is $\mathit{lc}_i$.
    We say that the invocation is \textit{completed} if there is a time after which the response of every invocation $\mathtt{read\_changes}(s_j)$ contains a change $\langle s_i, \mathit{lc}_i, s_j, * \rangle$.
\end{definition}

We define $\mathcal{C}_{s,t}$ as the set containing every change $c$ created for server $s$ such that the $\mathtt{reassign}$ operation led to the creation of $c$ is completed at time $t$.
It is straightforward to show that $\mathcal{C}_{s,t} \subseteq \mathcal{C}_{s,t'}$ for any server $s$ and $t \leq t'$, and we say that $\mathcal{C}_{s,t'}$ is \emph{more up-to-date} than $\mathcal{C}_{s,t}$ if $\mathcal{C}_{s,t} \subset \mathcal{C}_{s,t'}$.
Further, for each server $s$, we assume there is a change defining the initial weight of $s$.
Specifically, given $w$ as the initial weight of $s$, we assume that $\mathtt{reassign}(s,w)$ is completed at time $t=0$.
We denote the weight of a server $s$ at any time $t$ by $\mathtt{W}_{s,t}$, where $\mathtt{W}_{s,t} \triangleq \sum_{ \langle *, *, s, \Delta\rangle \in \mathcal{C}_{s,t} } \Delta$.
We also denote the weight of a set of servers $A \subseteq \mathcal{S}$ by $\mathtt{W}_{A,t}$, where $\mathtt{W}_{A,t} \triangleq \sum_{s \in A} \mathtt{W}_{s,t}$.
With these definitions, we are ready to define the weight reassignment problem.

\begin{definition}[Weight Reassignment Problem]\label{def:wra}
    Any algorithm that solves the weight reassignment problem must satisfy the following properties:
    \begin{itemize}
        \item Integrity. $\forall t \geq 0$, $\forall  F \subset \mathcal{S}$ such that $|F| = f$, $\mathtt{W}_{F,t} < \frac{\mathtt{W}_{\mathcal{S},t}}{2}$.
    
        \item Validity-I. When the $\mathtt{reassign}(s, \Delta)$ operation is completed, a change $\langle *,*,s,\Delta \rangle$ is created if Integrity is not violated; otherwise, a change $\langle *,*,s,0 \rangle$ is created.
            
        \item Validity-II. If $\mathtt{read\_changes}(s)$ is invoked at time $t$, a set containing $\mathcal{C}_{s,t}$ is returned as the response.    

        \item Liveness. If a correct server $s$ invokes $\mathtt{reassign}$ (resp. a correct process $p$ invokes $\mathtt{read\_changes}$), the invocation will eventually be completed, and $s$ (resp. $p$) will receive a message $\langle \mathtt{Complete},* \rangle$ (resp. a set of changes).
    \end{itemize}
\end{definition}

It is straightforward to see why the Liveness property is a part of the problem's definition.
In the following, we discuss why the other properties are required.
\begin{itemize}
    \item Integrity is a consequence of Property~\ref{pro:availability-wmqs}, which determines the relationship between the servers' weights and $f$, guaranteeing the system's availability over time.

    \item The second property states that a change must be created as the outcome of each $\mathtt{reassign}$ invocation.
    It also determines how the change must be created.
    
    Notice that Integrity might be violated if each invocation of $\mathtt{reassign}(s, \Delta)$ is completed by creating a change $\langle *,*,s, \Delta \rangle$ (see example below). 
    Hence, to avoid the violation of Integrity, an invocation $\mathtt{reassign}(s, \Delta)$ might be aborted, i.e., a change $\langle *,*,s,0 \rangle$ is created as the invocation's outcome.
    % Furthermore, notice that all invocations of $\mathtt{reassign}$ cannot be completed by creating changes with zero weights, as Validity-I enforces the invocations to create non-zero weights if there is no Integrity violation.
    
    \item The third property determines what responses to a $\mathtt{read\_changes}$ invocation are valid.
    Given any process that invokes $\mathtt{read\_changes}(s)$ at any time $t \geq 0$, it is clear that a valid response must be as up-to-date as $\mathcal{C}_{s,t}$.
    On the other hand, due to asynchrony, it is impossible to guarantee that exactly $\mathcal{C}_{s,t}$ is returned as the response.  
    Consequently, a valid response is one that contains $\mathcal{C}_{s,t}$.
\end{itemize}

\begin{exmp}\label{ex:1}
Let $\mathcal{S} = \{ s_1, s_2, s_3, s_4 \}$, $\Pi = \{ c_1, c_2 \}$, and $f = 1$.
The following sets contain the initial weights: 
    $\mathcal{C}_{s_1,0} = \{ \langle s_1, 1, s_1, 1 \rangle \}, \  
    \mathcal{C}_{s_2,0} = \{ \langle s_2, 1, s_2, 1 \rangle \}, \ 
    \mathcal{C}_{s_3,0} = \{ \langle s_3, 1, s_3, 1 \rangle \}, \text{ and } 
    \mathcal{C}_{s_4,0} = \{ \langle s_4, 1, s_4, 1 \rangle \}$.
Assume that server $s_1$ invokes $\mathtt{reassign}( s_1, 1.5 )$, which is completed at time $t_1$.
Accordingly, a change $\langle s_1, 2, s_1, 1.5 \rangle$ is created as the outcome of the invocation.
It is worth mentioning that this invocation cannot be completed by creating a change with zero weight, as Validity-I enforces the invocation to create a change with non-zero weight when there is no Integrity violation.

Client $c_1$ invokes $\mathtt{read\_changes}(s_1)$ after $t_1$ and receives a set of changes $\mathcal{C}$ at time $t_2$, where $\mathcal{C} = \mathcal{C}_{s_1, 0} \cup \{ \langle s_1, 2, s_1, 1.5 \rangle \}$.
Note that Validity-II is violated if $c_1$ receives $\mathcal{C}_{s_1, 0}$ (or any other set than $\mathcal{C}$).
Client $c_1$ can calculate the weight of $s_1$ using $\mathcal{C}$: the weight of $s_1$ equals $2.5$.
Server $s_3$ invokes $\mathtt{reassign}( s_2, -0.5 )$ after $t_2$, which is completed at time $t_3$.
Notice that creating a change $\langle s_3, 2, s_2, -0.5 \rangle$ violates Integrity, so a change $\langle s_3, 2, s_2, 0 \rangle$ is created.        
If client $c_2$ invokes $\mathtt{read\_changes}(s_2)$ after $t_3$, it will receive $\mathcal{C}' = \mathcal{C}_{s_2,0} \cup \{ \langle s_3, 2, s_2, 0 \rangle \}$.
It is important to note that servers are not allowed to invoke $\mathtt{reassign}( *, 0 )$ because the second parameter of $\mathtt{reassign}$ must be non-zero.
\end{exmp}

\section{Impossibility Result}\label{sec:impossibility}
We now show the weight reassignment problem introduced in the previous section\gs{instead of previous section, why not refer to the section (through a ref command) explicitly?} cannot be implemented in asynchronous failure-prone systems.
We start by presenting an insight into such an impossibility result.

Consider a system in which all correct servers invoke the $\mathtt{reassign}$ operation concurrently such that only one of the invocations can be completed by creating a change with non-zero weight.
That is, creating two or more changes, each with non-zero weight, violates Integrity, meaning that it can make $f$ servers have more than half of the total voting power in the system.
Assume that invocations $\mathtt{reassign}(s_1,\Delta_1), \dots, \mathtt{reassign}(s_n,\Delta_n)$ create such a situation.
One can take the following steps to solve consensus among servers:
\begin{enumerate}
    \item each correct server $s_i$ writes its proposal $v_i$ to a single-writer multi-reader ({\small SWMR}) register $R[i]$ and invokes $\mathtt{reassign}(s_i,\Delta_i)$, where $1\leq i\leq n$, and
    % \item if the created change with non-zero weight was proposed by $s_i$, the decided value is the one stored in $r_i$.
    \item if a change with non-zero weight is created for $s_j$, the decided value is the one stored in $R[j]$.
\end{enumerate}
Since the weight of only one of the created changes is non-zero, servers can decide the same value.
Consequently, consensus can be reduced to the weight reassignment problem, which means that the weight reassignment problem cannot be implemented in asynchronous failure-prone systems~\cite{wait-free}.

Based on this insight\gs{would this be an ``intuition'' (intuitions and science are not best friends) or an ``insight''?}, we design an algorithm presented in Algorithm~\ref{alg:cn:c-wr}, by which servers solve consensus using the weight reassignment problem, i.e., it reduces consensus to the weight reassignment problem.
The algorithm is executed by each correct server $s_i$ and provides a function -- $\mathtt{propose}(v_i)$ -- by which $s_i$ proposes a value $v_i$. 
We divide the servers into two disjoint sets, $F$ and $\mathcal{S} \setminus F$, such that $F=\{s_1,s_2,\dots,s_f\}$, and we assume that the initial weight of every server $s \in F$ (resp. $s \in \mathcal{S} \setminus F$) equals $\frac{n-1}{2f}$ (resp. $\frac{n+1}{2(n-f)}$).
Notice that Integrity is satisfied with these initial weights. 
Further, there is a shared array of {\small SWMR} registers $R$ of size $n$ to store servers' proposals.
% Further, the local counter of every server $s \in \mathcal{S}$ is $3$ at time $t$.

Each server $s_i$ executes the $\mathtt{propose}$ function.
After storing its proposal in $R[i]$, $s_i$ invokes $\mathtt{reassign}( s_i, 0.5 )$ (resp. $\mathtt{reassign}( s_i, -0.5 )$) if $s_i \in F$ (resp. $s_i \in \mathcal{S} \setminus F$).
It is straightforward to see that two or more invocations of $\mathtt{reassign}$ cannot be completed by creating changes with non-zero weights.
For instance, if $\mathtt{reassign}( s_1, 0.5 )$\gs{invocations of which call (issued) by $s_1$?} and $\mathtt{reassign}( s_{f+1}, -0.5 )$ are completed by creating changes $\langle s_1,2,s_1,0.5 \rangle$ and $\langle s_{f+1},2,s_{f+1},-0.5 \rangle$ at time $t>0$, then we have:
\begin{align*}
    & \mathtt{W}_{F,t} = f \times \frac{n-1}{2f}+0.5 = \frac{n}{2} \\
    & \frac{\mathtt{W}_{\mathcal{S},t}}{2} = \frac{\mathtt{W}_{F,t}+\mathtt{W}_{\mathcal{S}\setminus F,t}}{2} \\
    & \qquad = \frac{f \times \frac{n-1}{2f}+0.5 + (n-f)\times \frac{n+1}{2(n-f)}-0.5}{2} = \frac{n}{2},
\end{align*}
which means that Integrity is violated.

In a loop, for each server $s_j \in \mathcal{S}$, $s_i$ repeatedly invokes $\mathtt{read\_changes}(s_j)$ to see the invocation of which server\gs{similar as here before, I think that we normally use invocations of algorithms, call, procedures issues by a server/process.} is completed by creating a change with non-zero weight.
Because of Liveness, the loop will eventually terminate.
Assume that the invocation of server $s_j$ is completed by creating a change $\langle s_j, 2, s_j, \Delta\rangle$, where $\Delta\neq 0$.
Consequently, $s_i$ returns $R[j]$ as the decided value, and consensus among servers will be solved.

\begin{algorithm}[hbt!]
\caption{Reducing consensus to the weight reassignment problem
%, i.e., solving consensus using the weight reassignment problem
-- server $s_i$.}
\label{alg:cn:c-wr}

\begin{algorithmic}[1]
\STATEx{\hspace{-1.6em}$\triangleright$ $R$ is a shared array of {\small SWMR} registers with size $n$}
\STATEx{\hspace{-1.6em}$\triangleright$ if $i \in \{ 1,2,\dots f \}$, $\mathtt{W}_{s_i,0} = \frac{n-1}{2f}$; otherwise, $\mathtt{W}_{s_i,0} = \frac{n+1}{2(n-f)}$}
% \STATEx{\hspace{-1.6em}$\triangleright$ $\mathit{lc}_{i} = 3$ at time $t$}
\STATEx{\hspace{-1.6em}$\triangleright$ $s_i$ executes the $\mathtt{propose}$ function}

\vspace{0.3em}
\STATEx{\hspace{-1.6em}\textbf{function} $\mathtt{propose}(v_i)$}
    \STATE{$R[i] \leftarrow v_i$}

    \NoThen
    \IF{$i \in \{ 1,2, \dots, f \}$}                                                         
        \STATE{$\mathtt{reassign}( s_{i}, 0.5 )$}
    \ELSE
        \STATE{$\mathtt{reassign}( s_{i}, -0.5 )$}
    \ENDIF       

    \STATE{$\mathit{decided\_value} \leftarrow \perp$}

    \REPEAT 
        \NoDo
        \FOR{$j \in \{ 1,2,\dots n \}$}
            \NoThen
            \STATE{$\mathcal{C} \leftarrow \mathtt{read\_changes}(s_j)$}
            \IF{$\langle s_j,2,s_j, \Delta \rangle \in \mathcal{C}$ such that $\Delta \neq 0$}
                \STATE{$\mathit{decided\_value} \leftarrow R[j]$}
            \ENDIF
        \ENDFOR
    \UNTIL{$\mathit{decided\_value} \neq \perp$}
    
    \STATE{\textbf{return} $\mathit{decided\_value}$}                                      
\end{algorithmic}
\end{algorithm}

\begin{thm}\label{thm:cn:autonomous}
    Consensus can be reduced to the weight reassignment problem.
\end{thm}

\begin{cor}\label{cor:impossibility}
    The weight reassignment problem cannot be implemented in asynchronous failure-prone systems. 
\end{cor}

The proof of Theorem~\ref{thm:cn:autonomous} and other theorems presented in the following sections can be found in the appendix.

\section{Restricting the Weight Reassignment Problem}\label{sec:weak}
The weight reassignment problem presented in Section~\ref{sec:abstraction} cannot be implemented in asynchronous failure-prone systems according to Corollary~\ref{cor:impossibility}.
In this section, we try to restrict that problem in order to find a variant that can be implemented in such a system model.
We begin by keeping the total voting power constant by restricting the reassignments to be done in a pairwise manner.
We call the resulting problem the \textit{pairwise weight reassignment}.
 
\subsection{Pairwise weight reassignment} 
Avoiding the Integrity violation is the main difficulty in solving the weight reassignment problem, so we focus on this property.
% It is clear that:
% \begin{equation}\label{eq:s-sf-f-1}
%     \mathtt{W}_{\mathcal{S},t} = \mathtt{W}_{\mathcal{S} \setminus F,t} + \mathtt{W}_{F,t} \hfill (\forall F \subset \mathcal{S}, \forall t \geq 0)
% \end{equation}
% Using Integrity and Equation \ref{eq:s-sf-f-1} when $|F| = f$, it is straightforward to show that:
% \begin{equation}\label{eq:f-s}
%     \mathtt{W}_{F,t} < \frac{\mathtt{W}_{\mathcal{S},t}}{2} \hfill (\forall F \subset \mathcal{S} \text{ such that } |F|=f, \forall t \geq 0)
% \end{equation}
% Inequality \ref{eq:f-s}, which is equivalent to Integrity, states that the total weight of any $f$ servers must be less than half of the total weight of all servers.
Our first idea for restricting the problem is to ensure that the right-hand side of the inequality presented in the Integrity property and, consequently, the total weight of all servers, remains constant over time, i.e., at any time $t > 0$, $\mathtt{W}_{\mathcal{S},t} = \mathtt{W}_{\mathcal{S},0}$.
In this way, identifying Integrity violations might become easier because we just need to compare the total weight of the $f$ servers having the greatest weights with $\frac{\mathtt{W}_{\mathcal{S},0}}{2}$.

To apply this restriction, servers can reassign their weights in a pairwise manner, i.e., a server $s_j$ gains a weight $\Delta$ if and only if another server $s_i$ loses $\Delta$.
In such a situation, we say that $\Delta$ is \textit{transferred} from $s_{i}$ to $s_{j}$.
To represent this way of reassigning weights, we define a new operation as follows:

\begin{itemize}
    \item $\mathtt{transfer}(s_i,s_j,\Delta)$ that can be invoked by any server $s_k$ to transfer $\Delta \neq 0$ from $s_{i}$ to $s_{j}$.
\end{itemize}

Similar to the weight reassignment problem, processes can utilize the $\mathtt{read\_changes}$ operation to learn about changes created for each server.
When $\mathtt{transfer}(s_i,s_j,\Delta)$ invoked by $s_k$ is completed, two changes $c=\langle s_k,\mathit{lc}_{k},s_i,-\Delta' \rangle$ and $c'=\langle s_k,\mathit{lc}_{k},s_j,\Delta' \rangle$ corresponding to the transfer’s outcome are created, where $\Delta'$ is either $\Delta$ or $0$.
Further, a message $\langle \mathtt{Complete}, c \rangle$ is returned to $s_k$ (notice that both changes are created with either non-zero weights or zero weights, so returning $c$ is enough to determine the weight of $c'$).
We say that the transfer is completed if there is a time after which the responses of two invocations $\mathtt{read\_changes}(s_i)$ and $\mathtt{read\_changes}(s_j)$ contain $c$ and $c'$, respectively.
Each server increments its local counter after each $\mathtt{transfer}$ invocation. 
By convention, we say that a $\mathtt{transfer}$ invocation is \emph{effective} (resp. \emph{null}) if the weights of created changes are non-zero (resp. zero).

By considering this restriction, we define a new variant of the weight reassignment problem called the \textit{pairwise weight reassignment} in which $\mathtt{transfer}$ is the only operation to reassign weights.
The definition of the pairwise weight reassignment contains all properties of the weight reassignment problem (Definition~\ref{def:wra}) adapted to use the $\mathtt{transfer}$ operation instead of $\mathtt{reassign}$, as follows.

\begin{definition}[Pairwise Weight Reassignment]\label{def:pwr}
    Any algorithm that solves the pairwise weight reassignment problem must satisfy the following properties\gs{again I'd prefer something as, ``any algorithm solving the pairwise weight reassignment problem must satisfy the following properties: ''}:
    \begin{itemize}
        \item P-Integrity. $\forall t \geq 0$, $\forall  F \subset \mathcal{S}$ such that $|F| = f$, $\mathtt{W}_{F,t} < \frac{\mathtt{W}_{\mathcal{S},t}}{2}$.
        
        \item P-Validity-I. When the $\mathtt{transfer}(s_i, s_j, \Delta)$ operation is completed, two changes $\langle *,*,s_i,-\Delta \rangle$ and $\langle *,*,s_j, \Delta \rangle$ are created if P-Integrity is not violated; otherwise, two changes $\langle *,*,s_i, 0 \rangle$ and $\langle *,*,s_j, 0 \rangle$ are created.

        \item P-Validity-II. If $\mathtt{read\_changes}(s)$ is invoked at time $t$, a set containing $\mathcal{C}_{s,t}$ is returned as the response.
                    
        \item P-Liveness. If any correct server $s$ invokes $\mathtt{transfer}$ (resp. a correct process $p$ invokes $\mathtt{read\_changes}$), the invocation will  eventually be completed, and $s$ (resp. $p$) will receive a message $\langle \mathtt{Complete,*}\rangle$ (resp. a set of changes).
    \end{itemize}
\end{definition}

Now this question arises: \emph{Can the pairwise weight reassignment be implemented in asynchronous failure-prone systems?}
The answer to this question is \emph{no}.
The general idea behind this impossibility is similar to the one presented for the weight reassignment problem (Section~\ref{sec:impossibility}). 
Consider a set of servers $F\subset \mathcal{S}$ with size $f$, and assume that all correct servers invoke $\mathtt{transfer}$ concurrently such that only one of the transfers executed by members of $\mathcal{S} \setminus F$ can be completed effectively.
P-Integrity is indeed violated if two or more transfers executed by members of $\mathcal{S} \setminus F$ are completed effectively.
In such a situation, all correct servers can decide on the value proposed by a server $s \in \mathcal{S} \setminus F$ whose transfer is completed effectively (the decided value is selected from the values proposed by members of $\mathcal{S} \setminus F$.)
As a result, servers can solve consensus using pairwise weight reassignment, which means that consensus can be reduced to the pairwise weight reassignment problem.
Hence, pairwise weight reassignment cannot be implemented in asynchronous failure-prone systems.   

Based on this insight, we design an algorithm, presented in Algorithm~\ref{alg:pwr-impossibility}, to solve consensus using pairwise weight reassignment.
The algorithm is executed by each correct server $s_i$ and provides a function $\mathtt{propose}(v_i)$. 
We assume that the initial weight of each server $s\in F$ (resp. $s \in \mathcal{S}\setminus F$) is $\frac{n-1}{2f}$ (resp. $\frac{n+1}{2(n-f)}$), where $F = \{s_1,s_2,\dots,s_f\}$.
Further, there is a shared array of {\small SWMR} registers $R$ with size $n$ to store servers' proposals.
 % the local counter of every server is $3$ at time $t$.

Each server $s_i \in \mathcal{S}$ executes the $\mathtt{propose}$ function.
After storing its proposal in $R[i]$, each server $s_i$ invokes $\mathtt{transfer}$.
The transfers executed by servers must be in such a way that only one of the transfers executed by members of $\mathcal{S} \setminus F$ can be completed effectively.
To this end, each server $s_i \in F$ (resp. $s_i \in \mathcal{S}\setminus F$) invokes $\mathtt{transfer}( s_{i}, s_j, 0.1 )$ (resp. $\mathtt{transfer}( s_{i}, s_1, 0.4 )$), where $j = (i+1) \ \mathit{mod} \ f$.

In the loop, for each server $s_j \in \mathcal{S}\setminus F$, $s_i$ repeatedly invokes $\mathtt{read\_changes}(s_j)$ to see the transfer of which server is completed effectively.
Because of P-Liveness, the loop will eventually terminate.
Assume that the transfer of server $s_j$ is completed effectively by creating two changes $\langle s_j, 2, s_j, -0.4 \rangle$ and $\langle s_j, 2, s_1, 0.4 \rangle$.
Consequently, $s_i$ returns $R[j]$ as the decided value.

It is straightforward to see that the transfer of each correct server $s \in F$ completes effectively without changing the total weight of servers in $F$.
On the other hand, only one transfer executed by members of $\mathcal{S}\setminus F$ can be completed effectively; otherwise, P-Integrity is violated.
For instance, if transfers of $s_{f+1}$ and $s_{f+2}$ are completed effectively by creating changes $\langle s_{f+1},2,s_{f+1},-0.4 \rangle$, $\langle s_{f+1},2,s_1,0.4 \rangle$, $\langle s_{f+2},2,s_{f+2},-0.4 \rangle$, and $\langle s_{f+2},2,s_{1},0.4 \rangle$ at time $t>0$, then we have:
\begin{align*}
    & \mathtt{W}_{F,t} = f \times \frac{n-1}{2f}+0.4+0.4 = \frac{n}{2}+0.3 \\
    & \frac{\mathtt{W}_{\mathcal{S},t}}{2} = \frac{\mathtt{W}_{\mathcal{S},0}}{2} = \frac{n}{2},
\end{align*}
which means that P-Integrity is violated.

\begin{algorithm}[t!]
\caption{Reducing consensus to the pairwise weight reassignment problem
%, i.e., solving consensus using the pairwise weight reassignment 
-- server $s_i$.}
\label{alg:pwr-impossibility}

\begin{algorithmic}[1]
\STATEx{\hspace{-1.6em}$\triangleright$ $R$ is a shared array of {\small SWMR} registers with size $n$}
\STATEx{\hspace{-1.6em}$\triangleright$ if $i \in \{ 1,2,\dots f \}$, $\mathtt{W}_{s_i,0} = \frac{n-1}{2f}$; otherwise, $\mathtt{W}_{s_i,0} = \frac{n+1}{2(n-f)}$}
% \STATEx{\hspace{-1.6em}$\triangleright$ $\mathit{lc}_{i} = 3$ at time $t$}
\STATEx{\hspace{-1.6em}$\triangleright$ $s_i$ executes the $\mathtt{propose}$ function}

\vspace{0.3em}
\STATEx{\hspace{-1.6em}\textbf{function} $\mathtt{propose}(v_i)$}
    \STATE{$R[i] \leftarrow v_i$}

    \NoThen
    \IF{$i \in \{ 1,2, \dots, f \}$}
        \STATE{$j \leftarrow (i+1) \ \mathit{mod} \ f$}
        \STATE{$\mathtt{transfer}( s_i,s_j, 0.1 )$}
    \ELSE
        \STATE{$\mathtt{transfer}( s_i,s_1, 0.4 )$}
    \ENDIF       

    \STATE{$\mathit{decided\_value} \leftarrow \perp$}

    \REPEAT 
        \NoDo
        \FOR{$j \in \{ f+1,f+2,\dots n \}$}
            \NoThen
            \IF{$\langle s_j,2,s_1, 0.4 \rangle \in \mathtt{read\_changes}(s_j)$}
                \STATE{$\mathit{decided\_value} \leftarrow R[j]$}
            \ENDIF
        \ENDFOR
    \UNTIL{$\mathit{decided\_value} \neq \perp$}
    
    \STATE{\textbf{return} $\mathit{decided\_value}$}  
\end{algorithmic}
\end{algorithm}

\begin{thm}\label{thm:cn-pwr-n-f}
    Consensus can be reduced to the weight reassignment problem.
\end{thm}

% The proof of this theorem can be found in the appendix.
As the \gs{minor suggestion: you could eventually highlight quite briefly here why a pairwise approach (including $\mathtt{transfer}$) seemed to be a promising idea (that is worth considering.), even if it insufficient.}restriction presented above is insufficient to restrict the weight reassignment problem in a way that can be implemented in asynchronous failure-prone systems, we define another problem in the following.

\subsection{Restricted pairwise weight reassignment}
In pairwise weight reassignment, after invoking $\mathtt{transfer}$, servers must use consensus or similar primitives to create changes in order to preserve P-Integrity.
The objective of using consensus for each $\mathtt{transfer}$ invocation is to \emph{decide} whether the invocation is effective or not, i.e., which changes must be created: two changes with non-zero weights or with zero weights.
One possible approach to create changes without consensus is eliminating such a globally taken decision, i.e., given a server $s_i$ that wants to invoke $\mathtt{transfer}(*,*,\Delta)$, $s_i$ is allowed to execute the operation if its invocation does not violate (P-)Integrity.

We now present two conditions that, if they are satisfied, ensure an effective transfer\gs{consider rephrasing it; e.g., \dots two conditions that, if they are satisfied, can ensure(or simply ``ensure'') an effective transfer:}: 

\begin{itemize}
    \item (C1) only $s_i$ can invoke $\mathtt{transfer}(s_i,*,\Delta)$, i.e., other servers cannot transfer some of $s_i$'s weight, and
    \item (C2) the weight of $s_i$ must always be greater than $\frac{\mathtt{W}_{\mathcal{S},0}}{2(n-f)}$.
\end{itemize}

Note that if C1 holds, C2 is a locally verifiable condition.

\begin{thm}\label{thm:ensure-transfer}
    Provided that a server $s_i$ wants to invoke $\mathtt{transfer}$, we can ensure that the transfer is effective if C1 and C2 are met. 
\end{thm}

The proof of Theorem~\ref{thm:ensure-transfer} can be found in the appendix.
Here we present an insight into these conditions.
It is clear that:
\begin{equation}\label{eq:s-sf-f-1}
    \mathtt{W}_{\mathcal{S},t} = \mathtt{W}_{\mathcal{S} \setminus F,t} + \mathtt{W}_{F,t} \hfill (\forall F \subset \mathcal{S}, \forall t \geq 0)
\end{equation}
It is straightforward to obtain the following inequality using (P-)Integrity and Equation~\ref{eq:s-sf-f-1} when $|F| = f$:
\begin{equation}\label{eq:f-s}
    \mathtt{W}_{\mathcal{S} \setminus F,t} > \frac{\mathtt{W}_{\mathcal{S},t}}{2} \hfill (\forall F \subset \mathcal{S} \text{ such that } |F|=f, \forall t \geq 0)
\end{equation}
Inequality~\ref{eq:f-s}, which is equivalent to (P-)Integrity, states that the total weight of any $n-f$ servers must be greater than half of the total weight of all servers.
Notice that if the weight of each server is greater than $\frac{\mathtt{W}_{\mathcal{S},0}}{2(n-f)}$ at any time $t$, the total weight of servers in set $\mathcal{S} \setminus F$ is greater than $|\mathcal{S} \setminus F|\times \frac{\mathtt{W}_{\mathcal{S},0}}{2(n-f)} = \frac{\mathtt{W}_{\mathcal{S},0}}{2}$, i.e., $\mathtt{W}_{\mathcal{S} \setminus F,t} > \frac{\mathtt{W}_{\mathcal{S},0}}{2}$.
Hence, if C2 holds for each transfer, (P-)Integrity is always preserved.
To see why C1 is required, assume that at least two servers $s_i,s_k\neq s_j$ invoke $\mathtt{transfer}(s_j, *, \Delta_1)$ and $\mathtt{transfer}(s_j, *, \Delta_2)$ at time $t$ such that $\mathtt{W}_{s_j,t}-\Delta_1 -\Delta_2 \leq \frac{\mathtt{W}_{\mathcal{S},0}}{2(n-f)}$ but $\mathtt{W}_{s_j,t}-\Delta_1 > \frac{\mathtt{W}_{\mathcal{S},0}}{2(n-f)}$ and $\mathtt{W}_{s_j,t}-\Delta_2 > \frac{\mathtt{W}_{\mathcal{S},0}}{2(n-f)}$.
In fact, if both transfers are completed effectively, then C2 is violated; however, one of the transfers can be completed effectively without violating C2.
In such a situation, $s_i$ and $s_k$ can solve consensus (like the impossibility results presented in Sections~\ref{sec:impossibility} and \ref{sec:weak}.)
Consequently, in order to satisfy C2 in asynchronous failure-prone systems, we must assume that for each server $s_j$, there is at most one server that is allowed to invoke $\mathtt{transfer}(s_j, *, *)$.
Without loss of generality, we assume that only $s_j$ can invoke $\mathtt{transfer}(s_j, *, *)$.

These conditions indeed can restrict pairwise weight reassignment.
We define a new version of pairwise weight reassignment called \textit{restricted pairwise weight reassignment} to consider these conditions.
Specifically, it contains all properties of the pairwise weight reassignment (Definition~\ref{def:pwr}) except for two changes: P-Integrity is replaced by RP-Integrity to consider C2, and P-Validity-I is adapted so that only server $s$ can invoke $\mathtt{transfer}(s,*, *)$ due to C1.
In the next section\gs{minor comment: I'd prefer referring to the next section with the command ref}, we elaborate on how servers can use these conditions to transfer weights in asynchronous failure-prone systems.

\begin{definition}[Restricted Pairwise Weight Reassignment]\label{def:rpwr}
    Any algorithm that solves the restricted pairwise weight reassignment problem must satisfy the following properties:
    \begin{itemize}
        \item RP-Integrity. $\forall t \geq 0$, $\forall s \in \mathcal{S}$, $\mathtt{W}_{s,t} > \frac{\mathtt{W}_{\mathcal{S},0}}{2(n-f)}$. 
        
        \item RP-Validity-I. When the $\mathtt{transfer}(s_i, s_j, \Delta)$ operation is completed, two changes $\langle s_i,*,s_i,- \Delta \rangle$ and $\langle s_i,*,s_j, \Delta \rangle$ are created if RP-Integrity is not violated; otherwise, two changes $\langle s_i,*,s_i, 0 \rangle$ and $\langle s_i,*,s_j, 0 \rangle$ are created.

        \item RP-Validity-II. If $\mathtt{read\_changes}(s)$ is invoked at time $t$, a set containing $\mathcal{C}_{s,t}$ is returned as the response.
                    
        \item RP-Liveness. If any correct server $s$ invokes $\mathtt{transfer}$ (resp. a correct process $p$ invokes $\mathtt{read\_changes}$), the invocation will  eventually be completed, and $s$ (resp. $p$) will receive a message $\langle \mathtt{Complete,*}\rangle$ (resp. a set of changes).
    \end{itemize}
\end{definition}

% Note that for each $\mathtt{transfer}(s_i,*,*)$ invoked by a server $s_j$, two changes will be created such that the first elements of both changes is $s_j$.
% In RP-Validity-I, the first parameter of $\mathtt{transfer}$ must be the same as the first elements of the created changes.
% As a result, RP-Validity-I implies C1.

\begin{exmp}\label{ex:2}
    Let $\mathcal{S} = \{ s_{1}, s_{2}, s_{3}, s_{4}, s_5, s_6, s_7 \}$ and $f = 2$.
    In this setting, the weight of each server must be greater than $0.7$ at any time $t\geq 0$.
    Notice that the size of each quorum is four at the beginning.
    Assume that $\mathtt{transfer}$ is invoked by $s_4$, $s_5$, and $s_6$ according to Fig.~\ref{fig:example-rpwr}, and the invocations are completed before time $t_{1}$.
    The new weights of servers at time $t_1$ are presented in the figure.
    As a result of these weight reassignments, $\{ s_1,s_2,s_3 \}$ (a minority of servers) constitutes a quorum.

    This figure contains two other invocations made by $s_6$ and $s_7$ after time $t_1$.
    Notice that these invocations cannot be executed in the restricted pairwise weight reassignment due to RP-Integrity violation.
    However, they could be executed in the pairwise weight reassignment, resulting in the weights in the red shaded rectangular area.\gs{this could be rephrased in order to put emphasis on the new property, that prevents executing the invocations in the red shaded rectangular area.}
        
    \begin{figure}[t!]
        \centering
        \includegraphics[trim=0cm 4cm 0cm 4cm, clip, scale=.31]{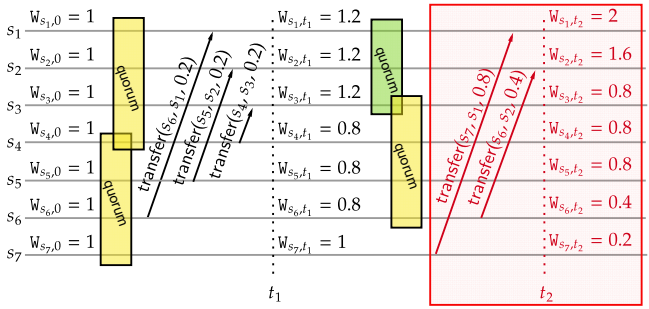}
        \caption{An example showing how the restricted pairwise weight reassignment works.
        The part surrounded by a red box cannot be executed in the restricted pairwise weight reassignment.}
        \label{fig:example-rpwr}
    \end{figure}
\end{exmp}   

\subsection{Discussion}
The restrictions imposed by the pairwise weight reassignment and restricted pairwise weight reassignment problems can lead to practical limitations in dealing with failed or slow servers in asynchronous systems. 
Here, we discuss these limitations in further detail.

In the weight reassignment problem, in the case of having a failed/slow server, there are two possible approaches to mitigate the impact of such a server: (I) decreasing the weight of the failed/slow server by other servers or (II) increasing the weights of other servers. 
This flexibility allows other servers to form quorums by a minority of servers during execution. 

However, in pairwise weight reassignment, the second approach cannot be employed in the case of having a failed/slow server, as the total weight of servers cannot change. 
The situation is even worse for the restricted pairwise weight reassignment, as servers cannot use both approaches when having a failed/slow server.
For instance, consider the same system as presented in Example \ref{ex:2}. 
Assume that the initial weights of servers $s_{1}, s_{2}, s_{3}, s_{4}, s_5, s_6,$ and $s_7$ are $1.6, 1.4, 0.8,0.8,0.8,0.8,$ and $0.8$, respectively. 
Assume that servers $s_1$ and $s_2$ are failed or slow. 
Then, the size of the smallest quorum is five, and servers cannot form smaller quorums by reassigning weights.
        
\section{Implementing the Restricted Pairwise Weight Reassignment}\label{sec:implem}
This section presents a protocol that implements the restricted pairwise weight reassignment in asynchronous failure-prone systems. 
The protocol is composed of two algorithms.
The first algorithm, Algorithm~\ref{alg:read-changes}, implements the $\mathtt{read\_changes}$ operation and contains two parts: the first part can be executed by any process (lines~\ref{line:function-read-changes}-\ref{line:return-end-changes}) and the second part must be executed by each correct server (lines~\ref{line:receipt-read-changes}-\ref{line:read-changes-ack}).  
To read the changes created for a server $s$, a process $p_i$ invokes $\mathtt{read\_changes}$.
Then, $p_i$ broadcasts a message $ \langle \mathtt{RC}, s, \mathit{lc}_i \rangle $ to all servers.
Upon receiving the message, each server responds by sending a set of changes it has stored for server $s$.
When $p_i$ receives more than $f$ responses, it takes the union of the received set of changes.
Let $\mathcal{C}$ be the resulting set.
Then, $p_i$ broadcasts $\mathcal{C}$ to all servers.
After receiving $\mathcal{C}$, each server stores $\mathcal{C}$ and responds by sending an acknowledgment.
As soon as receiving $n-f$ acknowledgments, $p$ can ensure that $\mathcal{C}$ is stored by at least $n-f$ servers and completes the invocation by returning $\mathcal{C}$.

\begin{algorithm}[ht!]
\caption{The implementation of $\mathtt{read\_changes}$.}
\label{alg:read-changes}
\begin{algorithmic}[1]
\STATEx{\hspace{-1.65em}$\triangleright$ This part can be executed by all process.}

% \vspace{0.5em}
\STATEx{\hspace{-1.6em}\textbf{operation} $\mathtt{read\_changes}(s)$}
    \STATE{$\mathcal{C} \leftarrow \emptyset $} \label{line:function-read-changes} 
    \STATE{\textbf{broadcast} $ \langle \mathtt{RC}, s \rangle $}
    \REPEAT\label{line:rc4}
        \STATE{\textbf{wait until} $\langle \mathtt{RC\_Ack}, \mathcal{C}_{s} \rangle$ is received from a server}  
        \STATE{$\mathcal{C} \leftarrow \mathcal{C} \cup \mathcal{C}_{s}$}
    \UNTIL{$|\mathcal{C}| > f $}\label{line:rc7}
    \STATE{\textbf{broadcast} $ \langle \mathtt{WC}, \mathcal{C} \rangle $}
    \STATE{\textbf{wait until} $\langle \mathtt{WC\_Ack} \rangle$ is received from $n-f$ servers}\label{line:rc9}
    \STATE{\textbf{return} $\mathcal{C}$}\label{line:return-end-changes}
\end{algorithmic}
% Question: Should clients check whether a change is completed?
\vspace{-0.7em}
\hrulefill
\vspace{-0.3em}
\begin{algorithmic}[1]    
\setcounter{ALG@line}{11}
    \STATEx{\hspace{-1.65em}$\triangleright$ This part is executed by every correct server $s_i$.}

    % \vspace{0.5em}
    \STATEx{\hspace{-1.65em}\textbf{upon receipt of} $\langle \mathtt{RC}, s \rangle$ from $p$} 
        \STATE{$\mathcal{C} \leftarrow \mathtt{get\_changes}(s)$ \hfill $\triangleright$ see Algorithm~\ref{alg:transfer} }\label{line:receipt-read-changes}
        \STATE{\textbf{send} $\langle \mathtt{RC\_Ack}, \mathcal{C}\rangle $ to $p$}

    \vspace{0.5em}
    \STATEx{\hspace{-1.65em}\textbf{upon receipt of} $\langle \mathtt{WC}, \mathcal{C} \rangle$ from $p$} 
        \STATE{$\mathtt{write\_changes}(\mathcal{C})$ \hfill $\triangleright$ see Algorithm~\ref{alg:transfer} }
        \STATE{\textbf{send} $\langle \mathtt{WC\_Ack}\rangle $ to $p$}\label{line:read-changes-ack}
\end{algorithmic}
\end{algorithm}

The second algorithm, Algorithm~\ref{alg:transfer}, contains the pseudo-code of the $\mathtt{transfer}$ operation and must be executed by each correct server.
It is worth mentioning that a server can invoke $\mathtt{transfer}$ only if its last invocation of $\mathtt{transfer}$ is complete.
We assume that servers invoke $\mathtt{transfer}$ based on the information provided by a monitoring system (to see how it can be implemented or its provided information can be used, refer to~\cite{aware, hheydari}).

The main idea behind Algorithm~\ref{alg:transfer} is as follows.
Each server $s$ can compute its weight using the changes stored in a local variable.
If its weight is greater than $\Delta + \frac{\mathtt{W}_{\mathcal{S},0}}{2(n-f)}$, it can transfer $\Delta$ to another server.
% Further, $s$ increases the local variable by $\Delta$ when an invocation $\mathtt{transfer}(*,s,\Delta)$ completes.

In further detail, each server has a local counter, denoted by $\mathit{lc}$, initialized with $1$, and used to distinguish the transferred weights.   
Also, each server has a variable denoted by $\mathit{register}$ to store the tag and the value of its local register (we elaborate on this variable in the next section.\gs{what about a ref to the section?}) 
If a server $s_{i}$ wants to transfer weight $\Delta$ to another server $s_{j}$, it first examines whether its weight remains greater than $\frac{\mathtt{W}_{\mathcal{S},0}}{2(n-f)}$ by transferring $\Delta$ (line~\ref{alg:dr:10})\gs{minor comment: perhaps ``line'' in lowercase?}.
If this is the case, then $s_{i}$ broadcasts a message $\langle \mathtt{T} , \langle s_{i}, \mathit{lc}_i, s_{i}, -\Delta \rangle, \langle s_{i}, \mathit{lc}_i, s_{j}, \Delta \rangle \rangle$ using a reliable broadcast primitive~\cite{Had94} (line~\ref{line:rb-transfer}).
Each server has a set denoted by $\mathcal{C}$ to store every received change, initialized with a set containing the initial weights of all servers.
By receiving a message $\langle \mathtt{T} , c,c' \rangle$ broadcast by $s_{i}$ (line~\ref{line:17}), every server adds $c$ and $c'$ to its set $\mathcal{C}$ and sends a response to $s_{i}$ (lines~\ref{line:10}-\ref{line:11}). 
If $s_{i}$ receives at least $n-f$ responses, then the invocation completes (lines~\ref{alg:dr:25}-\ref{alg:dr:26}).

\begin{algorithm}[hbt!]
\caption{The implementation of the $\mathtt{transfer}$ operation -- server $s_{i}$.}
\label{alg:transfer}
\begin{algorithmic}[1]
\STATEx{\hspace{-1.65em}\textbf{variables}}
    \STATE{$\mathit{lc}_i \leftarrow 1$}  
    \STATE{$\mathcal{C} \leftarrow \{ \langle s, 1, s, 1 \rangle \ | \ \forall s \in \mathcal{S} \}$}
    
     % \hspace*{-2.3em}\colorbox{gray!20}{\parbox{\linewidth}{
        \STATE{$\mathit{register}[\mathit{tag}[\mathit{ts}, \mathit{pid}], \mathit{val}] \leftarrow \langle \langle 0, \perp \rangle,\perp \rangle$}
     % }}
\vspace{0.3em}
\STATEx{\hspace{-1.65em}\textbf{function} $\mathtt{weight()}$}
    \STATE{$T \leftarrow \mathtt{get\_changes}(s_i)$}  
    \STATE{\textbf{return} $\sum_{\langle *,*,*,\Delta \rangle \in T} \Delta$}

\vspace{0.3em}
\STATEx{\hspace{-1.65em}\textbf{function} $\mathtt{get\_changes}(s)$}
    \STATE{\textbf{return} $\{ \langle *,*,s',* \rangle \ | \ \forall \ \langle *,*,s',* \rangle \in \mathcal{C}: s'=s \}$}

\vspace{0.3em}
\STATEx{\hspace{-1.65em}\textbf{function} $\mathtt{write\_changes}(\mathcal{C}')$}
    \NoThen
    \NoDo
    \FORALL{$\langle s_{j}, c, s_{k}, * \rangle \in \mathcal{C}' \setminus \mathcal{C}$}
        \IF{$i = k$}    \label{line:if-read}
            \STATE{$\mathit{register} \leftarrow \mathtt{read}()$} \label{line:read}
        \ENDIF
        \STATE{$\mathcal{C} \leftarrow \mathcal{C} \cup \{ \langle s_{j}, c, s_{k}, * \rangle \} $}\label{line:10} 
        \STATE{\textbf{send} $\langle \mathtt{T\_Ack}, c \rangle $ to $s_{j}$ if not already sent}\label{line:11} 
    \ENDFOR
    % \STATE{$\mathcal{C} \leftarrow \mathcal{C} \cup \mathcal{C}_{s}$}

\vspace{0.3em}
\STATEx{\hspace{-1.6em}\textbf{operation} $\mathtt{transfer}(s_i, s_{j}, \Delta)$}
    \NoThen
    \IF{$\mathtt{weight}() > \Delta + \frac{\mathtt{W}_{\mathcal{S},0}}{2(n-f)}$}    \label{alg:dr:10}
        \STATE{$\mathcal{C} \leftarrow \mathcal{C} \cup \{ \langle s_{i}, \mathit{lc}_i, s_{i}, -\Delta \rangle, \langle s_{i}, \mathit{lc}_i, s_{j}, \Delta \rangle \} $} 
        \STATE{\textbf{\small RB\_broadcast} $ \langle \mathtt{T} , \langle s_{i}, \mathit{lc}_i, s_{i}, -\Delta \rangle, \langle s_{i}, \mathit{lc}_i, s_{j}, \Delta \rangle \rangle $}\label{line:rb-transfer} 
        % \STATE{$acks \leftarrow \emptyset$}\label{alg:dr:25}
        % \REPEAT
        \STATE{\textbf{wait until} receiving $\langle \mathtt{T\_Ack}, \mathit{lc}_i \rangle $ from $n-f-1$ servers}\label{alg:dr:25}
            % \STATE{$\mathcal{C} \leftarrow \mathcal{C} \cup \mathcal{C}_{j}$}
            % \STATE{$acks \leftarrow acks \cup \{ s_{j} \}$}
        % \UNTIL{$|acks| \geq n-f$}\label{alg:dr:26}

        \STATE{$\mathit{msg} \leftarrow \langle \mathtt{Complete}, \langle s_{i}, \mathit{lc}_i, s_{i}, -\Delta \rangle \rangle $}
    \ELSE
        \STATE{$\mathit{msg} \leftarrow \langle \mathtt{Complete}, \langle s_{i}, \mathit{lc}_i, s_{i}, 0 \rangle \rangle $}
    \ENDIF
    \STATE{$\mathit{lc}_i \leftarrow \mathit{lc}_i + 1$}
    \STATE{\textbf{return} $\mathit{msg}$}\label{alg:dr:26}

\vspace{0.3em}
\STATEx{\hspace{-1.65em}\textbf{\small upon RB\_deliver} $\langle \mathtt{T} , \langle s_{j}, c, s_{j}, -\Delta \rangle, \langle s_{j}, c, s_{k}, \Delta \rangle \rangle$} 
    \STATE{$\mathtt{write\_changes}\left(\{ \langle s_{j}, c, s_{j}, -\Delta \rangle, \langle s_{j}, c, s_{k}, \Delta \rangle \}\right)$}\label{line:17}
    % \NoThen
    % \hspace*{-2.3em}\colorbox{gray!20}{\parbox{\linewidth}{
    % \IF{$i = k$}    \label{alg:dr:27}
    %     \STATE{$register \leftarrow \mathtt{read}()$}           
    %     % \STATE{$\mathcal{C} \leftarrow \mathcal{C} \cup \{ \langle s_{j}, c, s_{j}, -\Delta \rangle \} $} 
    %     % \STATE{\textbf{send} $\langle \mathtt{T\_Ack}, c\rangle $ to $s_{j}$} 
    % \ENDIF
    % % \ELSIF{$i\neq j$}
    %     \STATE{$\mathcal{C} \leftarrow \mathcal{C} \cup \{ \langle s_{j}, c, s_{j}, -\Delta \rangle, \langle s_{j}, c, s_{k}, \Delta \rangle \} $} 
    %     \STATE{\textbf{send} $\langle \mathtt{T\_Ack}, c \rangle $ to $s_{j}$} \label{alg:dr:30}
        % \STATE{\textbf{send} $ \langle s_{j}, c, s_{k}, \Delta \rangle $ to $s_{k}$}
    % \ENDIF

% \vspace{0.5em}
% \STATEx{\hspace{-1.65em}\textbf{\small upon receipt of} $c$ from $n-f-1$ servers}
%     \STATE{$register \leftarrow \mathtt{read}()$} 
%     \STATE{$\mathcal{C} \leftarrow \mathcal{C} \cup \{ c \} $}
\end{algorithmic} 
\end{algorithm}
    
\begin{thm}\label{thm:rpwr:satisfy-properties}
    The implementation of restricted pairwise weight reassignment (Algorithms~\ref{alg:read-changes} and \ref{alg:transfer}) satisfies the properties of the problem (Definition~\ref{def:rpwr}).
\end{thm}

\begin{thm}\label{thm:rpwr:implementation}
    Restricted pairwise weight reassignment can be implemented in asynchronous failure-prone systems.
\end{thm}
% The proofs of these theorems and other theorems presented in the following sections can be found in the extended version of the paper~\cite{awrmq}.
%full version of the paper~\cite{awrmq}.

\section{Dynamic-Weighted Atomic Storage}\label{ch:epoch-less-wr}
This section demonstrates how storage systems based on {\small MQS}  in asynchronous systems can be adapted to leverage the advantages of dynamic {\small WMQS}. 
To do so, we construct a dynamic-weighted atomic storage, where the weights of servers can be reassigned using the restricted pairwise weight reassignment, and its stored value can be accessed by two operations: $\mathtt{read}$ and $\mathtt{write}$.

In a nutshell, there are two main requirements to construct such storage.
First, each process $p$ (client or server) that wants to execute $\mathtt{read}$ or $\mathtt{write}$ protocols needs to store the most up-to-date set of the completed changes $\mathcal{C}$ that it knows. 
All $\mathtt{read}{/}\mathtt{write}$ protocol messages carry $\mathcal{C}$, and $p$ updates it as soon as discovering a more up-to-date set of completed changes. 
The servers reject any operation issued by $p$ containing a set of changes different from $\mathcal{C}$ and respond by sending their current set of completed changes to $p$, which updates its set $\mathcal{C}$. 
By receiving a set of changes that differs from its set of changes, $p$ restarts the executing operation. 
The second requirement is that, before accessing the system, $p$ must know the initial weight of each server. 

Our protocol extends the classical ABD algorithm \cite{abd} for supporting multiple writers and working with the 
 restricted pairwise weight reassignment. 
In the following, we highlight the main aspects of the $\mathtt{read}$ and $\mathtt{write}$ protocols (the complete algorithms can be found in the appendix.\gs{please do not forget to replace Appendix by the reference to the online document with proofs and additional information.}) 
These protocols work in phases. 
Each phase corresponds to accessing a weighted quorum of servers in $\mathcal{C}$. 
The $\mathtt{read}$ protocol works as follows:

\begin{itemize}
    \item \emph{1st Phase:} a reader requests a set of tuples  $\langle \mathit{tag},\mathit{val} \rangle$ ($\mathit{val}$ is the value a server stores, and $\mathit{tag}$ is its associated tag\footnote{A tag $\mathit{tg}$ is a pair holding the timestamp $\mathit{tg}.\mathit{ts}$ and the writer's id $\mathit{tg}.\mathit{pid}$ associated with the value stored by the register. A tag $\mathit{tg}_1$ is less than another $\mathit{tg}_2$ if (I) $\mathit{tg}_1.\mathit{ts} < \mathit{tg}_2.\mathit{ts}$, or (II) $\mathit{tg}_1.\mathit{ts} = \mathit{tg}_2.\mathit{ts}$ and $\mathit{tg}_1.\mathit{pid} < \mathit{tg}_2.\mathit{pid}$.}) from a weighted quorum of servers in the most up-to-date set of completed changes $\mathcal{C}$ and selects
    the one with the highest tag $\langle \mathit{tag}_h,\mathit{val}_h \rangle$;
    \item \emph{2nd Phase:} the reader performs an additional write-back phase in the system and waits for confirmations from a weighted quorum of servers in $\mathcal{C}$ before returning $\mathit{val}_h$.
\end{itemize}
    
The $\mathtt{write}$ protocol works in a similar way:
\begin{itemize}
    \item \emph{1st Phase:} a writer obtains a set of tags from a weighted quorum of servers in $\mathcal{C}$ and chooses the highest, $\mathit{tag}_h$; the tag to be written is defined by incrementing $\mathit{tag}_h.ts$ and assigning $\mathit{tag}_h.\mathit{pid}$ to the writer's id;
    \item \emph{2nd Phase:} the writer sends a tuple  $\langle \mathit{tag},\mathit{val} \rangle$  to the servers of $\mathcal{C}$, writing $\mathit{val}$ with tag $\mathit{tag}$, and waits for confirmations from a weighted quorum.
\end{itemize}

\vspace{0.5em}
\noindent\textbf{Correctness.} 
The following theorem states that the dynamic-weighted atomic storage can be implemented using the $\mathtt{read}{/}\mathtt{write}$ protocols if the weights of servers are reassigned by invoking $\mathtt{transfer}$ (Algorithm~\ref{alg:transfer}). 

\begin{thm}\label{thm:atomic-storage-correctness}
    The storage system implemented using the described $\mathtt{read}{/}\mathtt{write}$ protocols is atomic storage.
\end{thm}

\section{Related Work}
\label{sec:discussion}

\noindent\textbf{Weight (re)assignment.} 
The notion of \emph{majority} quorum was extended to be a \emph{weighted majority} quorum to improve the performance of replicated systems with diverse servers assigned with different voting power~\cite{weightedVoting}. 
{\small WHEAT} (WeigHt-Enabled Active replicaTion)~\cite{geoReplicatedSMR} shows that additional spare servers and weighted voting allow the system to benefit from diverse quorum sizes, enabling it to make progress by employing proportionally smaller quorums and potentially obtaining significant latency improvements.
In {\small WHEAT}, each assigned weight is either $w_{\mathit{min}}$ or $w_{\mathit{max}}$.
These values are defined in such a way that the safety and liveness properties of quorums are always satisfied.
We also considered a minimum weight for defining the restricted pairwise weight reassignment, just like {\small WHEAT}. 
Still, we consider more general weight schemes since our weights can have any value greater than the defined minimum.

In practice, network characteristics may be subject to run-time variations, and thus the assigned weights may also require to be changed over time. 
Accordingly, the problem of integrating {\small WMQS} with weight reassignment protocols was introduced to allow reassigning weights over time according to the observed performance variations.
This problem was studied for partially synchronous systems in~\cite{berger2023chasing, aware, dynVotingSchemeInDS, dynVotingAlgMaintainingConsistency, geoReplicatedSMR}, where the weight reassignment protocols are based on consensus or similar primitives.
Besides, such a problem was studied for asynchronous systems in~\cite{hheydari}, where the weights can be reassigned in a pairwise manner using an epoch-based protocol.
In the presented protocol, reassignment requests issued during an epoch can only be applied at the end of the epoch. 
Notice that the duration of epochs impacts the performance of the protocol, and\gs{it seems that there is no contrast between the previous sentence and this one. So that, you could introduce this sentence differently.} determining the optimal duration for epochs is challenging.
That study inspires our restricted pairwise weight reassignment; however, our implementation is epochless.
Moreover, in that study, the total weight of servers might become less than $\mathtt{W}_{\mathcal{S},0}$, leading to the loss of voting power as the system progresses.

\vspace{0.5em}
\noindent\textbf{Relationship with the asset transfer problem.} 
In the asset transfer problem, there are some accounts, each holding some assets owned by $k \geq 1$ servers. 
Some assets of any account can be transferred to another account if the source's balance does not become negative.
It was proved by Guerraoui et al.~\cite{cnCrypto}\gs{``It was proved by [author's name plus citation here]''} that if there is an account with $k > 1$ owners, the consensus number of the problem is $k$, i.e., the problem cannot be implemented in asynchronous failure-prone systems.
The insight into such an impossibility is as follows.
Consider an account with $k>1$ owner, and assume that all owners concurrently want to transfer some assets from the account to another account(s) such that the balance of the account will become negative by executing all transfers.
In such a situation, some transfers must be aborted to keep the account's balance non-negative, which requires consensus.

Reassigning weights in a pairwise manner is inspired by the asset transfer problem (consider weights equivalents to assets.)
These problems are similar in reassigning/transferring weights/assets, and both cannot be implemented in asynchronous failure-prone systems.
However, there is a significant difference between them: there is no condition related to the distribution of assets in the asset transfer problem, but the total of the $f$ greatest weights should be less than half of the total weights in the pairwise weight reassignment to satisfy P-Integrity.

To solve the asset transfer problem in asynchronous failure-prone systems, a restricted version of the problem called the $1$-asset transfer problem is presented in which each account is owned by exactly one owner.
It was proved that such a restricted problem could be implemented in asynchronous failure-prone systems~\cite{cnCrypto}.
In the restricted pairwise weight reassignment, the assumption that transferring a weight $\Delta$ from a server $s$ to another server can be made only by $s$ is inspired by the $1$-asset transfer problem.
% Table \ref{tbl:pwr-vs-1-at} describes the other similarities and differences between these problems.

% \input{tables/pairwise-vs-asset-transfer.tex}

\vspace{0.5em}
\noindent\textbf{Relationship with asynchronous reconfiguration.} 
Reconfigurable atomic storage \cite{dynAtomicStorageWithoutCons,effModConsensus-freeFST, jehl2017case, smartMerge, spiegelman2017dynamic} implements atomic registers~\cite{interprocessComI} in systems with the possibility of changing the set of servers over time, i.e., servers can join and leave the system during an execution.
The reconfigurable atomic storage is similar to the dynamic-weighted atomic storage in which quorum formations might change over time, i.e., a subset of servers that form a quorum during a time interval might not form a quorum after that interval.
Based on such a similarity, one might say that the techniques used to solve the reconfigurable atomic storage, e.g., generalized lattice agreement~\cite{faleiro2012generalized}, can be employed to solve the dynamic-weighted atomic storage; however, this is not the case because the system's availability\gs{system's availability?} is defined differently in these problems.

In dynamic-weighted atomic storage, the system remains available as long as the number of failures does not exceed the fault threshold; however, in reconfigurable atomic storage, the system remains available as long as any pending configuration has a majority of servers that did not crash and were not proposed for removal.    
In other words, the fault threshold is static and independent of the reassignment requests in the dynamic-weighted atomic storage; however, the fault threshold is dynamic and determined based on the pending join and leave requests in reconfigurable atomic storage (see Definition 1 in~\cite{dynAtomicStorageWithoutCons}).
% Further comparisons between these problems are presented in Table~\ref{tbl:comparing-reconfigurable-objects}. 
      
% \input{tables/atomic-storage.tex}
    
\section{Conclusion}
\label{sec:conclusion}

This paper studies the problem of integrating weighted majority quorums with weight reassignment protocols for any asynchronous system with a static set of servers and static fault threshold while guaranteeing availability.
We showed that such a problem could not be solved in asynchronous failure-prone distributed systems.
Then, we presented a restricted version of the problem called \textit{pairwise weight reassignment}, in which weights can only be reassigned pairwisely.
We showed that pairwise weight reassignment could not be implemented in asynchronous failure-prone systems.
We also discussed the relation between the pairwise weight reassignment and the asset transfer problem.
We presented a restricted version of the pairwise weight reassignment called \textit{restricted pairwise weight reassignment} that can be implemented in asynchronous failure-prone systems.
As a case study, we presented a dynamic-weighted atomic storage based on the implementation of the restricted pairwise weight reassignment.
%We also discussed the relationship between dynamic-weighted and reconfigurable atomic storage.
%Extending the failure model of this paper to Byzantine failures could be a direction for future work.

\section*{Acknowledgments}
We thank the ICDCS'23 anonymous reviewers for their constructive comments to improve the paper.
This work was supported by the Ministry of Higher Education and Research of France, FCT through the ThreatAdapt project (FCT-FNR/0002/2018) and the LASIGE Research Unit (UIDB/00408/2020 and UIDP/00408/2020), and by the European Commission through the VEDLIoT project (H2020 957197).

\bibliography{ref.bib}

\appendix
\noindent\textbf{Proof of Theorem~\ref{thm:cn:autonomous}}.
We show that servers can solve consensus using Algorithm~\ref{alg:cn:c-wr}, i.e., three properties of consensus -- Agreement, Validity, and Termination -- can be satisfied.
\begin{itemize}[leftmargin=1.67em]
    \item (Agreement) 
        Recall that each server $s$ executes the $\mathtt{propose}$ function in Algorithm~\ref{alg:cn:c-wr}.
        Then, $s$ invokes $\mathtt{reassign}$ according to its identifier.
        We first show that only one of the $\mathtt{reassign}$ invocations can be completed by creating a change with non-zero weight.
        For the sake of contradiction, assume that there is a set $A \subseteq \mathcal{S}$ such that  $|A| \geq 2$ and the invocation of each server $s \in A$ completes at time $t>0$ by creating a change with non-zero weight.
        Let $F = \{s_1, \dots, s_f \}$. 
        Assume that $k$ members of $A$ are in $F$, i.e., $|A \cap F| = k$.        
        We have:
        % \begin{equation}\label{eq:them:1:1}
        %     \mathtt{W}_{F,t} < \mathtt{W}_{\mathcal{S}\setminus F,t} \hfill (\forall F \subset \mathcal{S} \text{ such that } |F|=f, \forall t \geq 0)
        % \end{equation}
        
        % Since Inequality \ref{eq:them:1:1} holds at any time $t$, it holds at time $t'$ as well.
        % \begin{equation}\label{eq:them:1:11}
        %     \mathtt{W}_{F,t'} < \mathtt{W}_{\mathcal{S}\setminus F,t'} \hfill (\forall F \subset \mathcal{S} \text{ such that } |F|=f)
        % \end{equation}
        
        % Besides, we know that:
        % \begin{align}\label{eq:them:1:2}
        %        & \mathtt{W}_{F,t'} =  \underbrace{f \times \frac{n-1}{2f}}_{\mathtt{W}_{F,t}}  +  k \times 0.5 
        %     \\ & \mathtt{W}_{\mathcal{S}\setminus F,t'} = \underbrace{(n-f) \times \frac{n+1}{2(n-f)}}_{\mathtt{W}_{\mathcal{S}\setminus F,t}} - (|A|-k) \times 0.5 \nonumber
        % \end{align}
        \begin{align}\label{eq:them:1:2}
               & \mathtt{W}_{F,t} =  \underbrace{f \times \frac{n-1}{2f}}_{\mathtt{W}_{F,0}}  +  k \times 0.5 
            \\ & \mathtt{W}_{\mathcal{S}\setminus F,t} = \underbrace{(n-f) \times \frac{n+1}{2(n-f)}}_{\mathtt{W}_{\mathcal{S}\setminus F,0}} - (|A|-k) \times 0.5 \nonumber
        \end{align}

        It is straightforward to obtain the following inequality using Integrity and Inequality~\ref{eq:f-s}:
        \begin{equation}\label{eq:them:1:12345}
            \mathtt{W}_{F,t'} < \mathtt{W}_{\mathcal{S}\setminus F,t'} \hfill (\forall F \subset \mathcal{S} \text{ such that } |F|=f, \forall t' \geq 0)
        \end{equation}
        
        From Equations~\ref{eq:them:1:2} and Inequality~\ref{eq:them:1:12345}, we have:
        \begin{align*}
               & \mathtt{W}_{F,t'} <  \mathtt{W}_{\mathcal{S}\setminus F,t'} 
            \\ & \quad \Rightarrow \frac{n-1}{2}  +  k \times 0.5 < \frac{n+1}{2} - (|A|-k) \times 0.5 
            \\ & \quad \Rightarrow |A| < 2, 
        \end{align*}
        which is a contradiction since we assumed that $|A|\geq 2$.

        Next, we show that all invocations cannot be completed by creating changes with zero weights.
        For contradiction, assume all invocations are completed by creating changes with zero weights.
        There are two possibilities for a correct server $s_i$: $s_i\in F$ or $s_i \in \mathcal{S}\setminus F$.
        If $s_i\in F$, the invocation $\mathtt{reassign}(s_i,0.5)$ could be completed by creating a change with weight $0.5$, as Integrity is still preserved.
        Likewise, if $s_i\in \mathcal{S}\setminus F$, the invocation $\mathtt{reassign}(s_i,-0.5)$ could be completed by creating a change with weight $-0.5$.
        Since the invocation is completed by creating a change with zero weight, Validity-I is violated.

        Consequently, only one of the $\mathtt{reassign}$ invocations can be completed by creating a change with non-zero weight.
        Since the decided value corresponds to the invocation completed by creating a change with non-zero weight, the Agreement property is satisfied.
    \item (Validity)
        We say that a server $s$ is \emph{correct} if its $\mathtt{reassign}$ invocation is completed in Algorithm~\ref{alg:cn:c-wr}.
        We must show that if all correct servers propose the same value $v$, they must decide $v$.
        Note that the decided value in Algorithm~\ref{alg:cn:c-wr} is among the values proposed by servers whose $\mathtt{reassign}$ invocations are completed.
        In other words, the decided value is among the values proposed by the correct servers.
        Therefore, the Validity property is satisfied.
    \item (Termination)
        Notice that every invocation of the $\mathtt{reassign}$ operation will eventually terminate due to the Liveness property of the weight reassignment problem.
        Hence, the $\mathtt{reassign}$ invocation completed by creating a change with non-zero weight will eventually terminate.
        Also, the $\mathtt{read\_changes}$ invocations that enable servers to learn which $\mathtt{reassign}$ invocation is completed by creating a change with non-zero weight will eventually terminate.
        Consequently, each correct server can decide eventually.
\end{itemize}

\vspace{0.5em}
\noindent\textbf{Proof of Theorem~\ref{thm:cn-pwr-n-f}}. 
Like Theorem~\ref{thm:cn:autonomous}, we need to show that three properties of consensus -- Agreement, Validity, and Termination -- can be satisfied using Algorithm~\ref{alg:pwr-impossibility}.
Let $F = \{s_1, \dots, s_f \}$. 
\begin{itemize}
    \item (Agreement) 
        Recall that each server $s_i \in \mathcal{S}$ executes the $\mathtt{propose}$ function in Algorithm~\ref{alg:pwr-impossibility}.
        Then, $s_i$ invokes $\mathtt{transfer}$.
        Note that each server $s \in F$ transfers some of its weight to another member of $F$.
        Consequently, the total weight of members of $F$ does not change by completing transfers executed by members of $F$.
        We now show that only one of the transfers executed by members of $\mathcal{S} \setminus F$ can be completed by creating a change with non-zero weight.
        
        Like the proof of Theorem~\ref{thm:cn:autonomous}, we first show that multiple transfers executed by members of $\mathcal{S} \setminus F$ cannot be completed by creating changes with non-zero weights.
        For the sake of contradiction, assume that there is a set $A \subseteq \mathcal{S} \setminus F$ such that  $|A| \geq 2$ and the transfer of each server $s \in A$ completes at time $t>0$ by creating a change with non-zero weight.
        We know that:
        \begin{align}\label{eq:them:2:2}
               & \mathtt{W}_{F,t} = \underbrace{f \times \frac{n-1}{2f}}_{\mathtt{W}_{F,0}} + |A| \times 0.4
            \\ & \mathtt{W}_{\mathcal{S}\setminus F,t} = \underbrace{(n-f) \times \frac{n+1}{2(n-f)}}_{\mathtt{W}_{\mathcal{S}\setminus F,0} } - |A| \times 0.4 \nonumber
        \end{align}

        From Inequality~\ref{eq:them:1:12345} and Equations~\ref{eq:them:2:2}, we have:
        \begin{align*}
               & \mathtt{W}_{F,t} <  \mathtt{W}_{\mathcal{S}\setminus F,t} 
            \\ & \quad \Rightarrow \frac{n-1}{2}  +  |A| \times 0.4 < \frac{n+1}{2} - |A| \times 0.4 
            \\ & \quad \Rightarrow |A| < \frac{5}{4}, 
        \end{align*}
        which is a contradiction because $|A| \geq 2$ according to our assumption.

        Next, we show that all transfers executed by members of $\mathcal{S} \setminus F$ cannot be completed by creating changes with zero weights.
        For contradiction, assume all such transfers are completed by creating changes with zero weights.
        Consider a correct server $s_i \in \mathcal{S}\setminus F$.
        The invocation $\mathtt{transfer}(s_i,s_1,0.4)$ could be completed by creating two changes with weights $0.4$ and $-0.4$, because by creating such changes, P-Integrity is still preserved.
        Since the invocation is completed by creating changes with zero weights, P-Validity-I is violated.

        Consequently, only one of the transfers executed by members of $\mathcal{S} \setminus F$ can be completed by creating two changes with non-zero weights.
        Since the decided value corresponds to the transfer completed by creating two changes with non-zero weights and executed by a member of $\mathcal{S} \setminus F$, the Agreement property is satisfied.
    \item (Validity)
        This property holds by the same argument presented for the Validity property in the proof of Theorem~\ref{thm:cn:autonomous}.
    \item (Termination)
        Every $\mathtt{read\_changes}$ or $\mathtt{transfer}$ invocation will eventually terminate according to P-Liveness.
        Hence, the transfer executed by a member of $\mathcal{S} \setminus F$ and completed by creating a change with non-zero weight will eventually terminate.
        Besides, the $\mathtt{read\_changes}$ invocations that enable servers to learn the transfer of which member of $\mathcal{S} \setminus F$ is completed by creating a change with non-zero weight will eventually terminate.
        Consequently, each correct server $s$ can decide eventually.
\end{itemize}

\vspace{0.5em}
\noindent\textbf{Proof of Theorem~\ref{thm:ensure-transfer}}.
We present two preliminary lemmas before proving Theorem~\ref{thm:ensure-transfer}.

\begin{lem}\label{lem:wmin-p-integrity}
    If $\mathtt{W}_{s,t} > \frac{\mathtt{W}_{\mathcal{S},0}}{2(n-f)}$ for each server $s$ at any time $t$, then P-Integrity is always met.
\end{lem}
\begin{proof}
Recall that in the pairwise weight reassignment, the total weight of servers does not change during an execution, i.e., $\mathtt{W}_{\mathcal{S},t} = \mathtt{W}_{\mathcal{S},0}$ at any time $t>0$.
Also, recall that P-Integrity is equivalent to Inequality~\ref{eq:f-s}. 
Accordingly, we need to show that Inequality~\ref{eq:f-s} holds if $\mathtt{W}_{s,t} > \frac{\mathtt{W}_{\mathcal{S},0}}{2(n-f)}$ for each server $s$ at any time $t$.
We have:
\begin{align*}
    \mathtt{W}_{\mathcal{S} \setminus F,t} 
        &   = \sum_{s \in \mathcal{S} \setminus F} \mathtt{W}_{s,t} 
         > \sum_{s \in \mathcal{S} \setminus F} \frac{\mathtt{W}_{\mathcal{S},0}}{2(n-f)}
        \\& = |\mathcal{S} \setminus F| \times \frac{\mathtt{W}_{\mathcal{S},0}}{2(n-f)}
         = (n-f) \times \frac{\mathtt{W}_{\mathcal{S},0}}{2(n-f)} 
        \\& = \frac{\mathtt{W}_{\mathcal{S},0}}{2},
\end{align*}
which means that Inequality~\ref{eq:f-s} holds.
\end{proof}

\begin{lem}\label{lem:who-is-allowed}
    For each server $s$, there is at most one server that is allowed to invoke $\mathtt{transfer}(s,*,*)$ in order to preserve $\mathtt{W}_{s,t} > \frac{\mathtt{W}_{\mathcal{S},0}}{2(n-f)}$ at any time $t$ in asynchronous systems.
\end{lem}
\begin{proof}
    This proof is similar to the proof of Theorem~\ref{thm:cn:autonomous}.
    For contradiction, assume that $\mathtt{W}_{s,t} > \frac{\mathtt{W}_{\mathcal{S},0}}{2(n-f)}$ can be preserved at any time $t$ in asynchronous systems even if multiple servers invoke $\mathtt{transfer}(s,*,*)$.
    Particularly, assume that two correct servers $s_i,s_k\neq s_j$ invoke $\mathtt{transfer}(s_j, *, \Delta_1)$ and $\mathtt{transfer}(s_j, *, \Delta_2)$ at time $t$ such that $\mathtt{W}_{s_j,t}-\Delta_1 -\Delta_2 \leq \frac{\mathtt{W}_{\mathcal{S},0}}{2(n-f)}$ but $\mathtt{W}_{s_j,t}-\Delta_1 > \frac{\mathtt{W}_{\mathcal{S},0}}{2(n-f)}$ and $\mathtt{W}_{s_j,t}-\Delta_2 > \frac{\mathtt{W}_{\mathcal{S},0}}{2(n-f)}$.
    This means that only one of the transfers can be completed effectively, as if both transfers are completed effectively at time $t' > t$, then $\mathtt{W}_{s,t'} \leq \frac{\mathtt{W}_{\mathcal{S},0}}{2(n-f)}$.    
    % In such a situation, only one of the transfers executed by $s_i$ and $s_k$ can be completed effectively. 
    Therefore, $s_i$ and $s_k$ can decide on the value proposed by a server $s \in \{ s_i,s_k \}$ that its transfer is completed effectively, which is a contradiction since consensus cannot be solved in asynchronous systems. 
    Consequently, in order to preserve $\mathtt{W}_{s,t} > \frac{\mathtt{W}_{\mathcal{S},0}}{2(n-f)}$ in asynchronous systems, we must assume that for each server $s_j$, there is at most one server that is allowed to invoke $\mathtt{transfer}(s_j, *, *)$.
\end{proof}

Without loss of generality, we assume that only $s_j$ can invoke $\mathtt{transfer}(s_j, *, *)$ in Lemma~\ref{lem:who-is-allowed}.
Using Lemmas~\ref{lem:wmin-p-integrity} and \ref{lem:who-is-allowed}, the proof of Theorem~\ref{thm:ensure-transfer} is immediate.

% \subsection{Proofs related to restricted pairwise weight reassignment}

\input{appendix.tex}
\end{document}

%% file: extra-packages-commands.tex
\usepackage[leftmargin=3em]{quoting}
\usepackage{graphicx}
\usepackage[fleqn]{amsmath}
\usepackage{amssymb,amsfonts,amsthm}
\usepackage{mathtools}
\usepackage{mathrsfs}

\usepackage{caption}
\usepackage{subcaption}
\usepackage{algorithmicx}
\usepackage{algorithm}
\usepackage[noend]{algcompatible}
\usepackage[inline]{enumitem}
\usepackage{booktabs}
\usepackage{multicol}
\usepackage{lipsum}
\usepackage{color}
\usepackage{array}
\usepackage{tikz}

\theoremstyle{definition}

\newtheorem{property}{Property}

\newtheorem{definition}{Definition}
\newtheorem{thm}{Theorem}
\newtheorem{lem}{Lemma}
\newtheorem{cor}{Corollary}

\newtheoremstyle{exmp}% name
  {\topsep}% space above
  {\topsep}% space below
  {}% body font
  {}% indent amount
  {\bfseries}% theorem head font
  {.}% punctuation after theorem head
  {.5em}% space after theorem head
  {\thmname{#1}\thmnumber{ #2}\thmnote{ (#3)}}% theorem head spec
\theoremstyle{exmp}
\newtheorem{exmp}{Example}

\algdef{SE}[SUBALG]{Indent}{EndIndent}{ }{\algorithmicend\ }%
\algtext*{Indent}
\algtext*{EndIndent}
\newcommand\NoDo{\renewcommand\algorithmicdo{}}

\newcommand\NoThen{\renewcommand\algorithmicthen{}}

\usepackage{hyperref}
\hypersetup{
  colorlinks   = true,    % Colours links instead of ugly boxes
  urlcolor     = blue,    % Colour for external hyperlinks
  linkcolor    = blue,    % Colour of internal links
  citecolor    = red      % Colour of citations
}

\newboolean{showcomments}
\setboolean{showcomments}{false}

\ifthenelse{\boolean{showcomments}}
{ \newcommand{\mynote}[2]{
    \fbox{\bfseries\sffamily\scriptsize#1}
    {\small$\blacktriangleright$\textsf{\emph{#2}}$\blacktriangleleft$}}}
{ \newcommand{\mynote}[2]{}}

\newcommand{\gs}[1]{\mynote{GS}{#1}}

%% file: appendix.tex
\vspace{0.5em}
\noindent\textbf{Proof of Theorem~\ref{thm:rpwr:satisfy-properties}}.
We show that Algorithms~\ref{alg:read-changes} and \ref{alg:transfer} satisfy all properties of the restricted pairwise weight reassignment.
\begin{itemize}
    \item (RP-Integrity) Consider a server $s$ with a weight greater than $\Delta + \frac{\mathtt{W}_{\mathcal{S},0}}{2(n-f)}$.
    Assume that $s$ invokes $\mathtt{transfer}(s, *, \Delta)$ at time $t>0$.
    Since processes are sequential, it follows that $s$ cannot have any incomplete transfer at that time.
    When the transfer completes, $s$ adds a change $\langle s, \mathit{lc}, s, -\Delta \rangle$ to its local set $\mathcal{C}$, that contains the completed changes.
    As a consequence of adding such a change to $\mathcal{C}$, the weight of $s$ decreases by $\Delta$.
    Notice that before the transfer, the weight of $s$ was greater than $\Delta + \frac{\mathtt{W}_{\mathcal{S},0}}{2(n-f)}$, so its weight remains greater than $\frac{\mathtt{W}_{\mathcal{S},0}}{2(n-f)}$ when the transfer is completed.
    Since only $s$ can decrease its weight, it follows that RP-Integrity is always satisfied.
        
    \item (RP-Validity-I) Using Algorithm~\ref{alg:transfer}, a server $s$ can invoke $\mathtt{transfer}(s, *, \Delta)$ only when its weight is greater than $\Delta + \frac{\mathtt{W}_{\mathcal{S},0}}{2(n-f)}$.
    If it is the case, two changes will be created: $\langle s, \mathit{lc}, s, -\Delta \rangle$ and $\langle s, \mathit{lc}, *, \Delta \rangle$, and a message $\langle \mathtt{Complete}, \langle s, \mathit{lc}, s, -\Delta \rangle \rangle$ will be returned.
    Otherwise, a message $\langle \mathtt{Complete}, \langle s, \mathit{lc}, s, 0 \rangle \rangle$ will be returned without creating any change (zero-weight changes do not change the weights of servers, so it is not required to store them.)
    Consequently, RP-Validity-I is preserved.
    % event when $s$ does not create changes with zero weights when its weight is less than equal to $\Delta + \frac{\mathtt{W}_{\mathcal{S},0}}{2(n-f)}$.
    
    \item (RP-Validity-II)
        Assume that $\mathtt{read\_changes}(s)$ is invoked by a process $p$ (Algorithm~\ref{alg:read-changes}), where $s$ is a server, and $p$ receives a set $\mathcal{C}$ at time $t$ as the returned value.
        Further, assume that there is a change $c \in \mathcal{C}$ that is completed, i.e., $c \in \mathcal{C}_{s,t}$.
        We need to show that if any process $q$ that invokes $\mathtt{read\_changes}(s)$ after time $t$ will receive a set that contains $c$.

        For the sake of contradiction, assume that $q$ invokes $\mathtt{read\_changes}(s)$ after time $t$ and receives a set $\mathcal{C}'$ that does not contain $c$.
        When $q$ invokes $\mathtt{read\_changes}(s)$, the changes stored by at least $f+1$ servers must be collected (lines \ref{line:rc4}-\ref{line:rc7}).
        Since $c\notin \mathcal{C}'$, it is not stored by at least $f+1$ servers. 
        However, we know that before returning a value, it must be stored by at least $n-f$ servers in Algorithm~\ref{alg:read-changes} (line \ref{line:rc9}), and since $c\in \mathcal{C}$, it is stored by at least $n-f$ servers, which is a contradiction because there are $n$ servers in the system.
        
    \item (RP-Liveness) 
    Since at most $f$ servers might fail, there are $n - f$ correct servers in the worst case.
    Note that Algorithms~\ref{alg:read-changes} and \ref{alg:transfer} require no more than $n - f$ correct servers.
    Consequently, RP-Liveness holds.
\end{itemize}

\vspace{0.5em}
\noindent\textbf{Proof of Theorem~\ref{thm:rpwr:implementation}}.
According to Theorem~\ref{thm:rpwr:satisfy-properties}, Algorithms~\ref{alg:read-changes} and \ref{alg:transfer} implement restricted pairwise weight reassignment.
Both of these algorithms use primitives that can be implemented in asynchronous systems.
Also, operations in those algorithms require to contact with at most $n-f$ servers.
Since there are $n-f$ correct servers during an execution, any operation will eventually terminate.
Besides, Theorem~\ref{thm:atomic-storage-correctness} shows that any invocation of $\mathtt{read}$ operation will terminate eventually.
Consequently, any operation ($\mathtt{transfer}$ or $\mathtt{read\_changes}$) executed by a process can terminate eventually without waiting for the responses of other invocations of $\mathtt{transfer}$ or $\mathtt{read\_changes}$ operations, meaning that both Algorithms~\ref{alg:read-changes} and \ref{alg:transfer} can be implemented in asynchronous failure-prone systems.
It follows that restricted pairwise weight reassignment can be implemented in asynchronous failure-prone systems. 

% \subsection*{Dynamic-weighted atomic storage}
\vspace{0.5em}
\noindent\textbf{Reader-writer side of dynamic-weighted atomic storage.}
Algorithm~\ref{alg:client} is the pseudo-code of the $\mathtt{read}{/}\mathtt{write}$ protocols.
The $\mathtt{read}{/}\mathtt{write}$ protocols is similar to the $\mathtt{read}{/}\mathtt{write}$ protocols of the ABD protocol~\cite{abd} with only one difference: each reader or writer, after receiving messages from a set $Q \subseteq \mathcal{S}$ to decide whether a quorum is constituted, calls function $\mathtt{is\_quorum}$ (lines~\ref{alg:clien-line:isquorum-read} and \ref{alg:clien-line:isquorum-write} of Algorithm~\ref{alg:client}).

\begin{algorithm}[t!]
\caption{The reader-writer side of the $\mathtt{read}{/}\mathtt{write}$ protocols - process $p_{i}$.}
\label{alg:client}
\begin{algorithmic}[1]
    \STATEx{\hspace{-1.65em}\textbf{variables}}
        \STATE{$opCnt \leftarrow 0$}
        \STATE{$\mathcal{C} \leftarrow \{ \langle s, 1, s, 1 \rangle \ | \ \forall s \in \mathcal{S} \}$}
        
    \vspace{0.5em}
    \STATEx{\hspace{-1.65em}\textbf{functions}}
        \STATE{$\mathtt{read}() \equiv \mathtt{read\_write}(\perp)$}
        \STATE{$\mathtt{write}(value) \equiv \mathtt{read\_write}(value)$}
            
    \vspace{0.5em}
    \STATEx{\hspace{-1.65em}\textbf{function} $\mathsf{is\_quorum}( Q )$}  
        \NoThen
        \IF{$\frac{\mathtt{W}_{\mathcal{S},0}}{2} < sum\big(\{ \mathtt{W}_{s_i,*} \ | s_{i} \in Q\}\big)$}                                                         \label{alg:basics:common:line:start-quorum-constitution}
            % \STATEx{$\quad \textbf{or } \mathbb{n} - \mathbb{f} \leq |S|$}
            \STATE{\textbf{return} $\mathit{yes}$}
        \ELSE
            \STATE{\textbf{return} $\mathit{no}$}                                                                                 \label{alg:basics:common:line:end-quorum-constitution}
        \ENDIF 

    \vspace{0.5em}
    \STATEx{\hspace{-1.65em}\textbf{function} $\mathtt{read\_write}(value)$}
        \STATEx{\hspace{-1.3em}$\mathtt{phase1}$}
        \STATE{$opCnt \leftarrow opCnt + 1$}
        \STATE{\textbf{send} $\langle \mathtt{R}, opCnt \rangle$ to all servers} 
        \STATE{$Q \leftarrow  \emptyset$}
        \REPEAT
            \STATE{\textbf{upon receipt of} $\langle \mathtt{R\_A}, reg, opCnt, \mathcal{C}'\rangle$ from $s_{i}$}
            \Indent
                \vspace{-0.1cm}
                \IF{$\mathcal{C}\neq\mathcal{C}'$}
                    \STATE{$\mathcal{C} \leftarrow \mathcal{C}'$}
                    \STATE{$\mathtt{read\_write}(value)$ \hfill $\triangleright$ restart the operation} 
                \ENDIF
                \STATE{$Q \leftarrow Q \cup s_{i}.\langle reg, \mathcal{C} \rangle $}
            \EndIndent
        \UNTIL{$\mathtt{is\_quorum}(Q)$}\label{alg:clien-line:isquorum-read}
        \STATE{$maxtag \leftarrow \mathtt{max}\big(\{s_{i}.reg.\mathit{tag} \ | \ s_{i} \in Q \}\big)$}
        \STATE{$maxreg \leftarrow \mathtt{find}\big(\{s_{i}.reg \ | \ s_{i} \in Q \textbf{ and } $ }
        \STATEx{$\quad s_{i}.reg.\mathit{tag} = maxtag \}\big)$}
        \IF{$value = \perp$}
            \STATE{$value \leftarrow maxreg.value$}
        \ELSE
            \STATE{$ts \leftarrow maxtag.ts + 1$}
            \STATE{$\mathit{pid} \leftarrow p_{i}$}
        \ENDIF

        \STATEx{\hspace{-1.3em}$\mathtt{phase2}$}
        \STATE{\textbf{send} $\langle \mathtt{W}, \langle \langle ts,\mathit{pid} \rangle, value \rangle, opCnt \rangle$ to all servers}
        \STATE{$Q \leftarrow  \emptyset$}
        \REPEAT
            \STATE{\textbf{upon receipt of} $\langle \mathtt{W\_A}, reg, opCnt, \mathcal{C}'\rangle$ from $s_{i}$}
            \Indent
                \vspace{-0.1cm}
                \IF{$\mathcal{C}\neq\mathcal{C}'$}
                    \STATE{$\mathcal{C} \leftarrow \mathcal{C}'$}
                    \STATE{$\mathtt{read\_write}(value)$ \hfill $\triangleright$ restart the operation} 
                \ENDIF
                \STATE{$Q \leftarrow S \cup s_{i}.\langle reg, \mathcal{C} \rangle $}
            \EndIndent
        \UNTIL{$\mathtt{is\_quorum}(Q)$}\label{alg:clien-line:isquorum-write}
        \STATE{\textbf{return} $value$}
\end{algorithmic}
\end{algorithm}

\vspace{0.4em}
\noindent\textbf{Server side of dynamic-weighted atomic storage.}
Algorithm~\ref{alg:rw:server} is the pseudo-code of the servers' algorithm.
The servers' algorithm is similar to the servers' algorithm of the ABD protocol with only one difference: each server includes its set of changes, $\mathcal{C}$, to its responses.
                
\begin{algorithm}[t!]
\caption{The server side of the $\mathtt{read}{/}\mathtt{write}$ protocols - server $s_{i}$.}
\label{alg:rw:server}
\begin{algorithmic}[1]
    \STATEx{\hspace{-1.65em}\textbf{upon receipt of} $\langle \mathtt{R}, cnt \rangle$ from $p$}
        \STATE{\textbf{send} $\langle \mathtt{R\_A}, register, cnt, \mathcal{C}\rangle$ to $p$} \label{alg:basics:server:line:start-read}
    \vspace{0.5em}
    \STATEx{\hspace{-1.65em}\textbf{upon receipt of} $\langle \mathtt{W}, \langle \mathit{tag}, \mathit{val} \rangle, cnt\rangle$ from $p$}
        \NoThen
        \IF{$register.\mathit{tag} < \mathit{tag}$}
            \STATE{$register \leftarrow \langle \mathit{tag}, \mathit{val} \rangle$}
        \ENDIF
        \STATE{\textbf{send} $\langle \mathtt{W\_A}, cnt, \mathcal{C} \rangle$ to $p$}\label{alg:basics:server:line:end-write}
\end{algorithmic}
\end{algorithm}

\vspace{0.5em}
\noindent\textbf{Correctness}.
We show that our dynamic-weighted atomic storage  satisfies the safety and liveness properties of an atomic register according to the following definition:
\begin{definition}[Atomic register \cite{interprocessComII}]\label{def:atomic-register}
    Assume two read operations $r_{1}$ and $r_{2}$ executed by correct processes.
    Consider that $r_{1}$ terminates before $r_{2}$ initiates.
    If $r_{1}$ reads a value $\alpha$ from register $R$,
        then either $r_{2}$ reads $\alpha$ or $r_{2}$ reads a more up-to-date value than $\alpha$.
\end{definition}
% We prove that atomic storage can be implemented using the $\mathtt{read}{/}\mathtt{write}$ protocols (Algorithm~\ref{alg:client} and Algorithm \ref{alg:rw:server}) if the weights of servers are reassigned using Algorithm~\ref{alg:transfer}. 

\begin{lem}\label{lem:no-change}
    Assume there is no transfer from time $t_{1}$ to time $t_{2}$.
    Also, assume that set $S_{1} \subseteq \mathcal{S}$ (resp. $S_{2} \subseteq \mathcal{S}$) is determined as a quorum by function $\mathtt{is\_quorum}$ such that   every server $s_{1} \in S_{1}$ (resp. every server $s_{2} \in S_{2}$) is contacted from $t_{1}$ to $t_{2}$.
    Then, two quorums $S_{1}$ and $S_{2}$ have a non-empty intersection, i.e., $S_{1} \cap S_{2} \neq \emptyset$.
\end{lem}
\begin{proof}
    A set $S' \subseteq \mathcal{S}$ is determined as a quorum by function $\mathtt{is\_quorum}$ if the total weight of servers in $S'$ is greater than $\frac{\mathtt{W}_{\mathcal{S},0}}{2}$.
    For contradiction, assume that $S_{1} \cap S_{2} = \emptyset$.
    Sets $S_{1}$ and $S_{2}$ are determined as quorums.
    Let $tw_{1}$ and $tw_{2}$ be equal to the total weight of servers in $S_{1}$ and $S_{2}$, respectively.
    We have:
    \begin{align*}
        \begin{cases}
            \frac{\mathtt{W}_{\mathcal{S},0}}{2} < tw_{1} \\
            \frac{\mathtt{W}_{\mathcal{S},0}}{2} < tw_{2}
        \end{cases}
        \Rightarrow \mathtt{W}_{\mathcal{S},0} < tw_{1} + tw_{2}
    \end{align*}
    Since $S_{1} \cap S_{2} = \emptyset$, the total weight of servers $\{S_{1} \cup S_{2}\}$ is equal to $tw_{1} + tw_{2} > \mathtt{W}_{\mathcal{S},0}$; on the other hand, we know that $\{S_{1} \cup S_{2}\} \subseteq \mathcal{S}$ and the total weight of servers $\mathcal{S}$ is equal to $\mathtt{W}_{\mathcal{S},0}$; hence, we find a contradiction. 
\end{proof}

\begin{lem}\label{lem:one-change}
    Assume that a $\mathtt{read}$ operation $r_{1}$ returns $\langle \mathit{tag}_{\alpha}, \alpha \rangle$ at time $t_{\alpha}^{e}$.
    Also, assume that another $\mathtt{read}$ operation $r_{2}$ started at time $t_{\beta}^{s} > t_{\alpha}^{e}$ returns $\langle \mathit{tag}_{\beta}, \beta \rangle$.
    If there is only one transfer at time $t$ such that   $t_{\alpha}^{e} < t < t_{\beta}^{s}$, one of the following cases happen: (1) $\mathit{tag}_{\alpha} = \mathit{tag}_{\beta}$ and $\alpha = \beta$, or (2) $\mathit{tag}_{\alpha} \leq \mathit{tag}_{\beta}$ and $\beta$ was written after $\alpha$.
\end{lem}
\begin{proof}
    Without loss of generality, assume that the transfer increases the weight of a server $s_{i}$ and decreases the weight of a server $s_j$.
    Since there is only one transfer, the weights of other servers are not reassigned.
    Reassigning the weights of $s_{i}$ and $s_j$ is the only factor that causes the constitution of new quorums; since before accomplishing such a transfer, server $s_{i}$ executes a $\mathtt{read}$ operation to update its register (Algorithm~\ref{alg:transfer}, lines \ref{line:if-read}-\ref{line:read}), this lemma is similar to Lemma~\ref{lem:no-change}.
\end{proof}

\begin{lem}\label{lem:several-change}
    Assume that a $\mathtt{read}$ operation $r_{1}$ returns $\langle \mathit{tag}_{\alpha}, \alpha \rangle$ at time $t_{\alpha}^{e}$.
    Also, assume that another $\mathtt{read}$ operation $r_{2}$ started at time $t_{\beta}^{s} > t_{\alpha}^{e}$ returns $\langle \mathit{tag}_{\beta}, \beta \rangle$.
    If there is at least one transfer at time $t$ such that   $t_{\alpha}^{e} < t < t_{\beta}^{s}$, one of the following cases happen: (1) $\mathit{tag}_{\alpha} = \mathit{tag}_{\beta}$ and $\alpha = \beta$, or (2) $\mathit{tag}_{\alpha} \leq \mathit{tag}_{\beta}$ and $\beta$ was written after $\alpha$.
\end{lem}
\begin{proof}
    This lemma is straightforward using Lemma~\ref{lem:one-change}.
\end{proof}

\begin{lem}\label{lem:many-change}
    Assume that a $\mathtt{read}$ operation $r_{1}$ returns a value $\alpha_{1}$ at time $t_{1}^{e}$ with an associated tag $\mathit{tag}_{1}$.
    Also, assume that another $\mathtt{read}$ operation $r_{2}$ started at time $t_{2}^{s}>t_{1}^{e}$ returns a value $\alpha_{2}$ associated with tag $\mathit{tag}_{2}$.
    Then, one of the following cases happens: (1) $\mathit{tag}_{1} = \mathit{tag}_{2}$ and $\alpha_{1} = \alpha_{2}$, or (2) $\mathit{tag}_{1} \leq \mathit{tag}_{2}$ and $\alpha_{2}$ was written after $\alpha_{1}$.
\end{lem}
\begin{proof}
    Assume that $r_{1}$ starts at time $t_{1}^{s}$ and $r_{2}$ ends at time $t_{2}^{e}$, then $t_{1}^{s} < t_{1}^{e} < t_{2}^{s} < t_{2}^{e}$.
    There are four mutually exclusive cases:
    \begin{enumerate}[label=\alph*)]
        \item There is no concurrent transfer.
        For this case, the lemma holds according to Lemma~\ref{lem:no-change}.
        
        \item There is a concurrent transfer with $r_{1}$ at time $t'$ such that   $t_{1}^{s}< t' < t_{1}^{e}$.
        For this case, the lemma holds according to Lemma~\ref{lem:one-change}.
        
        \item There is at least one transfer between $r_{1}$ and $r_{2}$.
        For this case, the lemma holds according to Lemma~\ref{lem:several-change}.
        
        \item There is a concurrent transfer with $r_{2}$ at time $t'$ such that   $t_{2}^{s}< t' < t_{2}^{e}$.
        This case is similar to Case 2.       
    \end{enumerate}
\end{proof}

\vspace{0.4em}
\noindent\textbf{Proof of Theorem~\ref{thm:atomic-storage-correctness}}.
    Without loss of generality, we assume that this operation is invoked a finite number of times, like the asynchronous reconfiguration problem \cite{ReconfLatticeAgreement, spiegelman2017liveness}.
    According to Lemma~\ref{lem:many-change}, the storage system implemented using the $\mathtt{read}{/}\mathtt{write}$ protocols (Algorithm~\ref{alg:client} and Algorithm~\ref{alg:rw:server}) is atomic storage.
    Since at most $f$ servers might fail, there are $n - f$ servers in the worst case.
    The minimum value for the total weight of $n - f$ servers is greater than $\frac{\mathtt{W}_{\mathcal{S},0}}{2}$.
    Consequently, a quorum can be constituted, i.e., the system remains live even in the worst-case scenario.